%% file: main.tex
\documentclass[twoside,11pt]{article}

\input{preamble.tex}

\ShortHeadings{Nested Sequential Monte Carlo Methods}{Naesseth, Lindsten and Sch\"on}
\firstpageno{1}

\begin{document} 
\input{./coverArXiv}

\title{Nested Sequential Monte Carlo Methods}

\author{\name Christian A. Naesseth \email christian.a.naesseth@liu.se \\
        \addr Link\"oping University, Link\"oping, Sweden
        \AND
        \name Fredrik Lindsten \email fredrik.lindsten@eng.cam.ac.uk \\
        \addr The University of Cambridge, Cambridge, United Kingdom
        \AND
        \name Thomas B. Sch\"on \email thomas.schon@it.uu.se \\
        \addr Uppsala University, Uppsala, Sweden}
\editor{}

\maketitle

\begin{abstract}
  We propose \emph{nested sequential Monte Carlo} (\nsmc), a methodology to sample from sequences of
  probability distributions, even where the random variables are high-dimensional. %, \ie $X \in \reals^d, d\gg 1$.
  \nsmc generalises the \smc framework by requiring only approximate, \emph{properly weighted}, samples
  from the \smc proposal distribution, while still resulting in a correct \smc algorithm.
  Furthermore, \nsmc can in itself be used to produce such properly weighted samples. Consequently,
  one \nsmc sampler can be used to construct an efficient high-dimensional proposal distribution
  for another \nsmc sampler, and this \emph{nesting} of the algorithm can be done to an 
  arbitrary degree. This allows
  us to consider complex and high-dimensional models using \smc.
  We show results that motivate
  the efficacy of our approach on several
  filtering problems with dimensions in the order of 100 to \thsnd{1}.
\end{abstract} 


% Intro
\input{intro}

% Background/Problem formulation/IT-idea
\input{backgrd}

% Nested Idea
\input{nested}

% Nested SMC
\input{nsmc}

% Related work
\input{rw}

% Experiments
\input{expts}

\section{Conclusions}\label{sec:conclusions}
We have shown that a straightforward \nsmc implementation with fairly few particles can attain reasonable approximations to the filtering problem for dimensions %as high as $\thsnd{1}$, if not higher. 
in the order of hundreds, or even thousands.
This means that \nsmc methods takes the \smc framework an important step closer to being viable for high-dimensional statistical inference problems. However, \nsmc is not a silver bullet for solving high-dimensional inference problems, and the approximation accuracy will be highly model dependent.
Hence, much work remains to be done, for instance on combining \nsmc with other techniques for high-dimensional inference such as localisation \citep{rebeschiniH2015can}
and annealing \citep{beskosCJ2014on}, in order to solve even more challenging problems.

\section*{Acknowledgments} 
This work was supported by the projects: \emph{Learning of complex dynamical systems} (Contract
number: 637-2014-466) and \emph{Probabilistic modeling of dynamical systems} (Contract number:
621-2013-5524), both funded by the Swedish Research Council.

\clearpage
\appendix
\section{Appendix}
\input{suppMain}

\bibliography{refs}

\end{document}

%% file: preamble.tex
\usepackage{times}
\usepackage{graphicx} % more modern

\usepackage{natbib}

\usepackage{graphicx}      % include this line if your document contains figures
\usepackage{cite}
\usepackage{graphicx}
\usepackage{algorithmic,algorithm}
\usepackage{amssymb,amsmath,amsthm}

\usepackage{jmlr2e}
\usepackage{xspace}

\usepackage{color}

\usepackage{bbm}
\usepackage{wrapfig}
\usepackage{tikz}
\usetikzlibrary{arrows,decorations.pathmorphing,backgrounds,positioning,fit,petri}
\usepackage{array}
\usepackage{hyperref}
\usepackage{placeins}

\usepackage{lindsten,abbreviations}
\renewcommand\mid{\,\vert\,}

\newcommand\Zpi{Z_{\pi}}
\newcommand\Zq{Z_q}
\newcommand\setU{\mathsf{U}}
\newcommand\setXU{\Theta}

\newcommand{\target}{\bar\pi}        % Target distribution
\newcommand{\utarget}{\pi}             % Unnormalized target distribution

\newcommand{\extTarget}{\bar\Pi}
\newcommand{\extUTarget}{\Pi}
\newcommand{\extQ}{\bar Q}
\newcommand{\extQarg}[2]{\bar Q_{#2}^M(x_{#2}, u_{#2} \mid x_{1:#1}, u_{#1}) }
\newcommand{\WghtIS}{\omega}
\newcommand{\WghtSISR}{W'}

\newcommand{\Ztarget}[1]{Z_{\pi_{#1}}}        % Normalization constant for the target distribution
\newcommand{\myd}{\mathrm{d}}              % Differential for integrals etc.

\newcommand{\Zprop}[1]{Z_{q_{#1}}}        % Normalization constant for the proposal distribution

\newcommand\XX{\mathbf{X}}
\renewcommand\AA{\mathbf{A}}
\newcommand\UU{\mathbf{U}}
\newcommand\xx{\mathbf{x}}
\renewcommand\aa{\mathbf{a}}
\newcommand\uu{\mathbf{u}}

\newcommand\dSMC{\bar\Psi_{\text{\nsmc}}}
\newcommand\dBS[1]{\bar\Psi_{\text{BS,#1}}}
\newcommand\dPMCMC{\bar\Phi}

\newcommand\Class{\mathsf{Q}}
\newcommand\Obj[1][q]{\mathsf{#1}}

\newcommand\GetZ{\mathsf{GetZ}}
\newcommand\Simulate{\mathsf{Simulate}}

\newcommand\addr{\#\mathsf{h}}

\newcommand\eqdef{:=}
\newcommand\defeq{=:}
\newcommand\pdf{PDF\@\xspace}
\newcommand\pf{PF\@\xspace}

\newcommand\mrf{MRF\@\xspace}
\newlength\Papproxplotheight
\newlength\boxplotheight
\makeatletter
\newcounter{class}

\makeatother

\newcounter{asmp}
\def\asmpfont{\upshape}
\newenvironment{asmp}{\refstepcounter{asmp}\par\trivlist
   \item[\hskip \labelsep{\bfseries(A\theasmp)}]\asmpfont}
   {\endtrivlist}

\newcommand\asmpref[1]{(A\ref{#1})}

%% file: coverArXiv.tex
\newcommand{\coverTitle}{Nested Sequential Monte Carlo Methods}
\newcommand{\coverAuthors}{Christian A. Naesseth, Fredrik Lindsten and Thomas B. Sch{\"o}n}
\newcommand{\coverYear}{2015}
\newcommand{\coverStatus}{Accepted for publication.}

\begin{titlepage}
\begin{center}
{\large \em Technical report}

\vspace*{2.5cm}
%
%% TITLE
{\Huge \bfseries \coverTitle  \\[0.4cm]}

%
%% AUTHORS
{\Large \coverAuthors \\[2cm]}

\renewcommand\labelitemi{\color{red}\large$\bullet$}
\begin{itemize}
\item {\Large \textbf{Please cite this version:}} \\[0.4cm]
\large
\coverAuthors. \coverTitle. In \textit{Proceedings of the
$\mathit{32}$nd International Conference on Machine Learning},
Lille, France, 2015. JMLR: W\&CP volume 37.
%\begin{verbatim}
%@techreport{naessethls2015nested,
  %author    = {Christian A. Naesseth and 
               %Fredrik Lindsten and 
               %Thomas Sch\"on},
  %title     = {Nested Sequential Monte Carlo Methods},
  %journal   = {arXiv},
  %year      = {2015},
%}
%\end{verbatim}

\end{itemize}
%{\em \coverStatus}
\vfill

\begin{abstract}
  We propose \emph{nested sequential Monte Carlo} (\nsmc), a methodology to sample from sequences of
  probability distributions, even where the random variables are high-dimensional. %, \ie $X \in \reals^d, d\gg 1$.
  \nsmc generalises the \smc framework by requiring only approximate, \emph{properly weighted}, samples
  from the \smc proposal distribution, while still resulting in a correct \smc algorithm.
  Furthermore, \nsmc can in itself be used to produce such properly weighted samples. Consequently,
  one \nsmc sampler can be used to construct an efficient high-dimensional proposal distribution
  for another \nsmc sampler, and this \emph{nesting} of the algorithm can be done to an 
  arbitrary degree. This allows
  us to consider complex and high-dimensional models using \smc.
  We show results that motivate
  the efficacy of our approach on several
  filtering problems with dimensions in the order of 100 to \thsnd{1}.
\end{abstract}

\begin{keywords}
high-dimensional inference, high-dimensional particle filter, exact approximation, optimal proposal, sequential Monte Carlo, importance sampling, spatio-temporal models
\end{keywords}

\vfill

\end{center}
\end{titlepage}

%% file: intro.tex
\section{Introduction}
\label{sec:intro}

Inference in complex and high-dimensional statistical models is a very challenging problem that is
ubiquitous in applications. Examples include, but are definitely not limited to, climate informatics
\citep{monteleoniEtAl2013}, bioinformatics \citep{cohen2004bioinformatics} and machine learning
\citep{wainwright2008graphical}.
In particular, we are interested in \emph{sequential} Bayesian inference, which involves computing integrals
of the form
\begin{align}
  \label{eq:expectation}  
  \target_k(f) \eqdef \E_{\target_k}[f(X_{1:k})] = \int f(x_{1:k}) \target_k(x_{1:k}) \myd x_{1:k},
\end{align}
for some sequence of probability densities
\begin{align}
  \label{eq:target-def}
  \target_k (x_{1:k}) &= \Ztarget{k}^{-1} \utarget_k (x_{1:k}) ,  &k &\geq 1,
\end{align}
with normalisation constants $\Ztarget{k} = \int \utarget_k(x_{1:k}) \myd x_{1:k}$. Note that
$x_{1:k} \eqdef \prange{x_1}{x_k} \in \setX_k$.
The typical scenario that we consider is the well-known problem of inference in time series or
state space models \citep{ShumwayS:2011,cappeMR2005inference}. Here the index $k$
corresponds to time and we want to process some \emph{observations} $y_{1:k}$ in a sequential manner
to compute expectations with respect to the filtering distribution
$\target_k (\myd x_k) = \Prb(X_k \in \myd x_k \mid y_{1:k})$. 
To be specific, we are interested in settings where
\vspace{-2mm}
\begin{enumerate}
  \setlength{\itemsep}{-1mm}
\item[\emph{(i)}] $X_k$ is high-dimensional, \ie $X_k \in \reals^d$ with $d \gg 1$, and
\item[\emph{(ii)}] there are \emph{local dependencies} among the latent variables $X_{1:k}$, both \wrt
  time~$k$ and between the individual components of the (high-dimensional) vectors $X_k$.
\end{enumerate}

One example of the type of models we consider are the so-called spatio-temporal models
\citep{Wikle:2015,CressieW:2011,RueH:2005}. In Figure~\ref{fig:droughtmodel} we provide a probabilistic
graphical model representation of a spatio-temporal model that we will explore
further in Section~\ref{sec:expts}.

\begin{figure}[tb]
\centering
\resizebox{.5\columnwidth}{!} {
\input{3dmodel.tex}
}
\caption{Example of a spatio-temporal model where $\target_k (x_{1:k})$ is described by a $k \times 2 \times 3$ undirected graphical model and $x_k \in \reals^{2\times 3}$.}\label{fig:droughtmodel}
\end{figure}
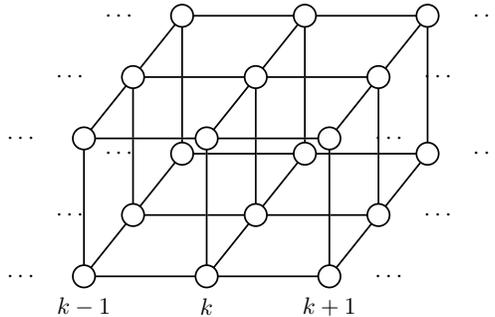

Sequential Monte Carlo (\smc) methods, reviewed in Section~\ref{sec:backgrd:smc},
comprise one of the most successful methodologies for sequential Bayesian inference.
However, \smc struggles in high-dimensions and these methods are rarely used for dimensions, say, $d \geq 10$ \citep{rebeschiniH2015can}.
The purpose of the \nsmc methodology 
is to push this limit well beyond $d = 10$.

The basic strategy, described in Section~\ref{sec:Backg:AdaptProp},
is to mimic the behaviour of a so-called \emph{fully adapted} \smc algorithm.
Full adaptation can drastically improve the efficiency of \smc in high dimensions. Unfortunately, it can rarely
be implemented in practice since the fully adapted proposal distributions are typically intractable.
\nsmc addresses this difficulty by requiring only approximate, \emph{properly weighted}, samples
from the proposal distribution. The proper weighting condition ensures the validity of \nsmc,
thus providing a generalisation of the family of \smc methods.
Furthermore, \nsmc will itself produce properly weighted samples. Consequently,
it is possible to use one \nsmc procedure within another to construct efficient high-dimensional proposal distributions.
This \emph{nesting} of the algorithm can be done to an arbitrary degree. For instance, for the model depicted
in Figure~\ref{fig:droughtmodel} we could use three nested samplers, one for each dimension of the ``volume''.

The main methodological development is concentrated to Sections~\ref{sec:nested}--\ref{sec:nsmc}.
We introduce the concept of proper weighting, approximations of the proposal distribution, and
nesting of Monte Carlo algorithms.
Throughout Section~\ref{sec:nested} we consider simple importance sampling
and in Section~\ref{sec:nsmc} we extend the development to the sequential setting.

We deliberately defer the discussion of the existing body of related work until
Section~\ref{sec:rw}, to open up for a better understanding of the relationships to the new
developments presented in Sections~\ref{sec:nested}--\ref{sec:nsmc}. We also discuss various attractive features of \nsmc that are of interest in high-dimensional settings, \eg the
fact that it is easy to distribute the computation, which results in improved memory efficiency and
lower communication costs. Section~\ref{sec:expts}
profiles our method extensively with a state-of-the-art competing algorithm %\fl{It's only one of the methods that can be referred to as state of the art (for high-dim). Should we rephrase this?} 
on several high-dimensional data sets. We also show the performance of inference and the modularity of the method
on a
$d=1\thinspace 056$ dimensional climatological spatio-temporal model \citep{fuBLS2012drought} structured according
to Figure~\ref{fig:droughtmodel}. Finally, in Section~\ref{sec:conclusions} we conclude the paper with some final remarks.

%%% Local Variables:
%%% mode: latex
%%% TeX-master: "main.tex"
%%% End:

%% file: 3dmodel.tex
%transforms all coordinates the same way when used (use it within a scope!)
%(rotation is not 45 degress to avoid overlapping edges)
% Input: point of origins x and y coordinate
\newcommand{\myGlobalTransformation}[2]
{
    \pgftransformcm{1}{0}{0.4}{0.5}{\pgfpoint{#1cm}{#2cm}}
}

\tikzstyle myBG=[thick,opacity=1.0]

% draws lines with white background to show which lines are closer to the
% viewer (Hint: draw from bottom up and from back to front)
%Input: start and end point
\newcommand{\drawLinewithBG}[2]
{
    \draw[white,myBG]  (#1) edge (#2);
    \draw[black,thick] (#1) edge (#2);
}

% draws all horizontal graph lines within grid
\newcommand{\graphLinesHorizontal}
{
    \drawLinewithBG{1,1}{5,1};
    \drawLinewithBG{1,3}{5,3};
    \drawLinewithBG{1,5}{5,5};
    
    \drawLinewithBG{1,1}{1,5};
    \drawLinewithBG{3,1}{3,5};
    \drawLinewithBG{5,1}{5,5};
}

% draws all vertical graph lines within grid
\newcommand{\graphLinesVertical}
{
    %swaps x and y coordinate (hence vertical lines):
    \pgftransformcm{0}{1}{1}{0}{\pgfpoint{0cm}{0cm}}
    \graphLinesHorizontal;
}

%draws nodes of the grid
%Input: point of origins x and y coordinate
\newcommand{\graphThreeDnodes}[2]
{
    \begin{scope}
        \myGlobalTransformation{#1}{#2};
        \foreach \y in {1,3,5} {
			\node at (0,\y) [circle] {$\cdots$};
        }
        \foreach \y in {1,3,5} {
			\node at (1,\y) [draw,circle,thick,fill=white] {};
			%this way circle of nodes will not be transformed
        }
        \foreach \y in {1,3,5} {
			\node at (3,\y) [draw,circle,thick,fill=white] {};
			%this way circle of nodes will not be transformed
        }
        \foreach \y in {1,3,5} {
			\node at (5,\y) [draw,circle,thick,fill=white] {};
			%this way circle of nodes will not be transformed
        }
        \foreach \y in {1,3,5} {
			\node at (6,\y) [circle] {$\cdots$};
        }
    \end{scope}
}

\begin{tikzpicture}[>=triangle 60]

    %draws helper-grid:
    %\gridThreeD{0}{0}{black!50};
    %\gridThreeD{0}{4.25}{black!50};
        % draws all graph nodes:

    %draws lower graph lines and those in z-direction:
    \begin{scope}
        \myGlobalTransformation{0}{0};
        \graphLinesHorizontal;

        %draws all graph lines in z-direction (reset transformation first!):
        \foreach \x in {1,3,5} {
            \foreach \y in {1,3,5} {
                \node (thisNode) at (\x,\y) {};
                {
                    \pgftransformreset
                    \draw[white,myBG]  (thisNode) -- ++(0,2.25);
                    \draw[black,thick] (thisNode) -- ++(0,2.25);
                }
            }
        }
    \end{scope}

    %draws upper graph-lines:
    \begin{scope}
        \myGlobalTransformation{0}{2.25};
        \graphLinesVertical;
    \end{scope}

	\graphThreeDnodes{0}{0};
    \graphThreeDnodes{0}{2.25};
	\node at (1.4,0) [circle] {$k-1$};
	\node at (3.4,0) [circle] {$k$};
	\node at (5.4,0) [circle] {$k+1$};
\end{tikzpicture}

%%% Local Variables:
%%% mode: latex
%%% TeX-master: "main.tex"
%%% End:

%% file: backgrd.tex
\section{Background and Inference Strategy}\label{sec:backgrd}

%========================================================================
%===============================  Sequential Monte Carlo  ===============================
%========================================================================
\subsection{Sequential Monte Carlo}\label{sec:backgrd:smc}
Evaluating $\target_k(f)$ as well as the normalisation constant $\Ztarget{k}$ in \eqref{eq:target-def} is typically
intractable and we need to resort to approximations. \smc methods, or particle filters (\pf), constitute a
popular class of numerical approximations for sequential inference problems. Here we give a
high-level introduction to the concepts underlying \smc methods, and postpone the details to
Section~\ref{sec:nsmc}.  For a more extensive treatment we refer to 
\citet{doucetJ2011a,cappeMR2005inference,doucetDG2001an}. In particular, we will use the auxiliary
\smc method as proposed by \citet{pittS1999filtering}.

At iteration $k-1$, the \smc sampler approximates the target distribution $\target_{k-1}$ by a collection of weighted particles (samples) $\{(X_{1:k-1}^i, W_{k-1}^i)\}_{i=1}^N$. These samples define an empirical
point-mass approximation of the target distribution
\begin{equation}
\target_{k-1}^N(\myd x_{1:k-1}) \eqdef \sum_{i=1}^N \frac{W^i_{k-1}}{\sum_\ell W^\ell_{k-1}}
\delta_{X_{1:k-1}^i}(\myd x_{1:k-1}),
\end{equation}
where $\delta_{X}(\myd x)$ denotes a Dirac measure at $X$. Each iteration of the \smc
algorithm can then conceptually be described by three steps, resampling, propagation, and weighting.

The resampling step puts emphasis on the most promising particles by discarding
the unlikely ones and duplicating the likely ones. The propagation and weighting steps essentially correspond
to using importance sampling when changing the target distribution from~$\target_{k-1}$
to~$\target_k$, \ie simulating new particles from a \emph{proposal distribution}
and then computing corresponding importance weights.

\subsection{Adapting the Proposal Distribution}
\label{sec:Backg:AdaptProp}
The first working \smc algorithm was the bootstrap \pf by \citet{GordonSS:1993}, which propagates
particles by sampling from the system dynamics and computes importance weights according to the
observation likelihood (in the state space setting). However, it is well known that the bootstrap
\pf suffers from weight collapse in high-dimensional settings \citep{bickelLB2008sharp}, \ie the
estimate is dominated by a single particle with weight close to one.  This is an effect of the
mismatch between the importance sampling proposal and the target distribution, which typically gets
more pronounced in high dimensions.

More efficient proposals, partially alleviating the degeneracy issue for some models, can be designed by \emph{adapting} the proposal distribution to the target
distribution (see Section~\ref{sec:nsmc:fapf}).
In \citet{naessethls2014capacity} we make use of the \emph{fully adapted} \smc method
\citep{pittS1999filtering} for doing inference in a (fairly) high-dimensional \emph{discrete}
model where $x_k$ is a $60$-dimensional discrete vector.
We can then make use of forward filtering and backward simulation, operating on the individual
\emph{components} of each $x_k$, in order to sample from the fully adapted \smc proposals. However, this method is limited to models where the latent space is either
discrete or Gaussian and the optimal proposal can be identified with a tree-structured graphical model.
Our development here can be seen as a non-trivial extension of this technique.  Instead of coupling one
\smc sampler with an \emph{exact} forward filter/backward simulator (which in fact reduces to an
instance of standard \smc), we derive a way of coupling multiple \smc samplers and \smc-based
backward simulators. This allows us to construct procedures for
mimicking the efficient fully adapted proposals for arbitrary latent spaces and structures in
high-dimensional models.

%%% Local Variables:
%%% mode: latex
%%% TeX-master: "main.tex"
%%% End:

%% file: nested.tex
\section{Proper Weighting and Nested Importance Sampling}\label{sec:nested}

In this section we will lay the groundwork for the derivation of the class of \nsmc algorithms. 
We start by considering the simpler case of importance sampling (\is), which is a fundamental component of \smc,
and introduce the key concepts that we make use of.
In particular, we will use a (slightly nonstandard) presentation of an algorithm as an instance of
a \emph{class}, in the object-oriented sense, and show that these classes can be nested to an arbitrary degree.

\subsection{Exact Approximation of the Proposal Distribution}
Let $\target(x) = \Zpi^{-1}\utarget(x)$ be a target distribution of interest.
\is can be used to estimate an expectation $\target(f) \eqdef \E_{\target}[f(X)]$ by sampling from a proposal distribution
$\bar q (x) = \Zq^{-1}q(x)$ and computing the estimator $ (\sum_{i=1}^\Np W^i)^{-1} \sum_{i=1}^\Np W^i f(X^i)$, with $W^i
= \frac{Z _q\pi(X^i)}{q(X^i)}$, and
where $\{(X^i,W^i)\}_{i=1}^N$ are the weighted samples. It is possible to replace
the \is weight by a nonnegative unbiased estimate, and still obtain a valid (consistent, \etc.) algorithm \citep[p. 37]{liu2001monte}.
One way to motivate this approach is by considering the random weight to be an auxiliary variable
and to extend the target distribution accordingly. Our development is in the same flavour,
but we will use a more explicit condition on the relationship between the random weights and the simulated particles.
Specifically, we will make use of the following key property to formally justify the proposed algorithms.

\begin{definition}[Properly weighted sample]
  A (random) pair $(X, W)$ is properly weighted for an \emph{unnormalised} distribution $p$ if $W \geq 0$ and
  $\E[ f(X) W ] = p(f) \eqdef \int f(x)p(x)\myd x$ for all measurable functions $f$.
\end{definition}

Note that proper weighting of $\{(X^i, W^i)\}_{i=1}^N$ implies unbiasedness of the
estimate of the normalising constant of~$p$.
Indeed, taking $f(x) \equiv 1$
gives
\( \E\left[ \frac{1}{N} \sum_{i=1}^N W^i \right] = \int p(x) \myd x =: Z_p \).

Interestingly, to construct a valid \is algorithm for our target $\target$ it is sufficient to generate samples that are properly weighted \emph{\wrt the proposal} distribution $q$.
To formalise this claim, assume that we are not able to simulate exactly from $\bar q$,
but that it is possible to evaluate the unnormalised density $q$ point-wise.
Furthermore, assume we have access to a class $\Class$, which works as follows.
The constructor of $\Class$ requires the specification of an \emph{unnormalised density function}, say, $q$,
which will be approximated by the procedures of $\Class$.
Furthermore, to highlight the fact that we will typically use \is (and \smc) to construct $\Class$,
the constructor also takes as an argument a precision parameter $\Mp$, corresponding to the number of samples used by the
``internal'' Monte Carlo procedure. An object is then instantiated as $\Obj = \Class(q,\Mp)$. The class $\Class$ is assumed to have the following properties:
\begin{asmp}
  \label{asmp:proper}
  Let $\Obj = \Class(q,\Mp)$. Assume that:
  \vspace{-2mm}
  \setlength{\itemsep}{-1mm}
  \begin{enumerate}
  \item The construction of $\Obj$ results in the generation
    of a (possibly random) member variable, accessible as $\widehat Z_q = \Obj.\GetZ()$.
    The variable $\widehat Z_q$ is a nonnegative, unbiased estimate of the normalising
    constant $Z_q = \int q(x)dx$.
  \item $\Class$ has a member function $\Simulate$ which returns a (possibly random) variable
    $X = \Obj.\Simulate()$, such that $(X, \widehat Z_q)$ is properly weighted for $q$.
  \end{enumerate}
\end{asmp}

With the definition of $\Class$ in place, it is possible to generalise\footnote{With $\Obj.\GetZ() \mapsto Z$ and $\Obj.\Simulate()$ returning a sample from $\bar q$ we obtain the standard \is method.}
the basic importance sampler as in Algorithm~\ref{alg:nsmc:nested-is}, which generates weighted samples $\{ (X^i, W^i)  \}_{i=1}^\Np$ targeting $\target$.
Note that Algorithm~\ref{alg:nsmc:nested-is} is different from a random weight IS, since it
approximates the proposal distribution (and not just the importance weights).

\begin{algorithm}
  \caption{Nested \is (steps 1--3 for $i = \range{1}{\Np}$)}
  \label{alg:nsmc:nested-is}
  \begin{enumerate}
  \item Initialise $\Obj^i = \Class(q, M)$.
  \item Set $\widehat Z_q^i = \Obj^i.\GetZ()$ and $X^i = \Obj^i.\Simulate()$.
  \item Set $W^i = {\displaystyle \frac{ \widehat Z_q^i \pi(X^i)}{ q( X^i) }}$.
  \item Compute ${\widehat Z}_{\pi} = \frac{1}{\Np} \sum_{i=1}^\Np W^i.$
  \end{enumerate}
\end{algorithm}

To see the validity of Algorithm~\ref{alg:nsmc:nested-is} we can interpret the sampler as a standard \is
 algorithm for an extended target distribution, defined as
$\extTarget(x,u) \eqdef u \,\extQ(x,u) \target(x) q^{-1}(x)$,
where $\extQ(x,u)$ is the joint \pdf of the random pair $(\Obj.\Simulate(), \Obj.\GetZ())$.
 Note that
$\extTarget$ is indeed a \pdf that admits $\target$ as a marginal; 
for any measurable subset $A\subseteq \setX$,
\begin{align*}
  &\extTarget(A \times \reals_+)
  = \int \mathbbm{1}_A(x) \frac{u\,\target(x)}{q(x)}\extQ(x,u) \myd x\myd u 
  = \E\left[\widehat Z_q \frac{\mathbbm{1}_A(X) \target(X)}{q(X)} \right] = \bar q\left( \mathbbm{1}_A\frac{\target}{ q} \right) \Zq = \target(A),
\end{align*}
where the penultimate equality follows from the fact that $(X, \widehat Z_q)$ is properly weighted for $q$.
Furthermore, the standard unnormalised \is weight for a sampler with target $\extTarget$ and proposal $\extQ$ is given by $u \,\utarget/q$, in agreement with Algorithm~\ref{alg:nsmc:nested-is}.

Algorithm~\ref{alg:nsmc:nested-is} is an example of what is referred to as an \emph{exact approximation};
see \eg, \citet{AndrieuR:2009,andrieuDH2010particle}.
Algorithmically, the method appears to be an approximation of an \is, but samples
generated by the algorithm nevertheless target the correct distribution $\target$. 

\subsection{Modularity of Nested \is}\label{sec:nested:modular}
To be able to implement Algorithm~\ref{alg:nsmc:nested-is} we need to define a class $\Class$
with the required properties \asmpref{asmp:proper}. The modularity of the procedure (as well as its name)
comes from the fact that we can use Algorithm~\ref{alg:nsmc:nested-is} also in this respect.
Indeed, let us now view $\target$---the target distribution of Algorithm~\ref{alg:nsmc:nested-is}---as the
\emph{proposal distribution} for another Nested \is procedure and consider the following definition of $\Class$:
\vspace{-2mm}
\begin{enumerate}
  \setlength{\itemsep}{-1mm}
\item Algorithm~\ref{alg:nsmc:nested-is} is executed at the construction of the object $\Obj[p] = \Class(\utarget, \Np)$,
  and $\Obj[p].\GetZ()$ returns the normalising constant estimate $\widehat Z_{\pi}$.
\item $\Obj[p].\Simulate()$ simulates a categorical random variable $B$ with $\Prb(B = i) = W^i / \sum_{\ell=1}^\Np W^\ell$
  and returns $X^B$.
\end{enumerate}\vspace{-0.5\baselineskip}

Now, for any measurable $f$ we have,\vspace{-0.5\baselineskip}
\begin{multline}
  \E[ f(X^B) \widehat Z_\pi ] = \sum_{i=1}^{\Np} \E \left[ f(X^i) \widehat Z_\pi \frac{ W^i }{ N \widehat Z_\pi } \right] 
  = \frac{1}{N} \sum_{i=1}^{\Np} \E \left[ f(X^i) \frac{\widehat Z_q^i \utarget(X^i)}{ q(X^i) } \right] \nonumber \\
  = \bar q\left( \frac{f\utarget}{q} \right)\Zq = \target(f)\Zpi,\nonumber
\end{multline}
where, again, we use the fact that $(X^i, \widehat Z_q^i)$ is properly weighted for $q$.
This implies that $(X^B, \widehat Z_\pi)$ is properly weighted for $\pi$ and that our definition of $\Class(\pi, N)$
indeed satisfies condition \asmpref{asmp:proper}.

The Nested \is algorithm in itself is unlikely to be of direct practical interest. However,
in the next section we will, essentially, repeat the preceding derivation in the context of \smc
to develop the \nsmc method.

%%% Local Variables:
%%% mode: latex
%%% TeX-master: "main.tex"
%%% End:

%% file: nsmc.tex
\section{Nested Sequential Monte Carlo}\label{sec:nsmc}
\subsection{Fully Adapted SMC Samplers}
Let us return to the sequential inference problem. As before, let
\(
\target_k (x_{1:k}) = \Ztarget{k}^{-1}\utarget_k (x_{1:k})
\)
denote the target distribution at ``time'' $k$. The unnormalised density $\utarget_k$ can be evaluated point-wise,
but the normalising constant $\Ztarget{k}$ is typically unknown.
We will use \smc to simulate sequentially from the distributions $\{\target_k \}_{k=1}^n$.
In particular, we consider the fully adapted SMC sampler \citep{pittS1999filtering},
which corresponds to a specific choice of resampling weights and proposal distribution, chosen in
such a way that the importance weights are all equal to $1/N$. Specifically, the proposal distribution (often referred
to as the \emph{optimal proposal}) is given by
$ \bar q_k(x_{k} \mid x_{1:k-1}) = \Zprop{k}(x_{1:k-1})^{-1} q_k (x_k \mid x_{1:k-1})$,
where
\begin{align*}
  q_k (x_k \mid x_{1:k-1}) &\eqdef \frac{\utarget_k(x_{1:k}) }{ \utarget_{k-1}(x_{1:k-1}) }.
\end{align*}
In addition, the normalising ``constant'' $\Zprop{k}(x_{1:k-1}) = \int q_k (x_k \mid x_{1:k-1}) \myd x_k$
is further used to define the \emph{resampling weights}, \ie the particles at time $k-1$ are
resampled according to $\Zprop{k}(x_{1:k-1})$ before they are propagated to time $k$.
For notational simplicity, we use the convention
$x_{1:0} = \emptyset$,
 $q_1(x_1 \mid x_{1:0}) = \utarget_1(x_1)$ and $Z_{q_1}(x_{1:0}) = \Ztarget{1}$.
The fully adapted auxiliary \smc sampler is given in Algorithm~\ref{alg:nsmc:fapf}.

\begin{algorithm}[tb]
  \caption{SMC (fully adapted)}
  \label{alg:nsmc:fapf}
  \begin{enumerate}
  \item 
    Set $\widehat Z_{\pi_0} =~1$.
  \item \textbf{for $k=1$ to $n$}
    \begin{enumerate}
    \item Compute $\widehat Z_{\pi_k} = \widehat Z_{\pi_{k-1}}\times\frac{1}{\Np}\sum_{j=1}^\Np \Zprop{k}(X_{1:k-1}^j).$
    \item  Draw $m_k^{1:\Np}$ from a multinomial distribution with probabilities
      \(\frac{\Zprop{k}(X_{1:k-1}^j)}{ \sum_{\ell=1}^\Np \Zprop{k}(X_{1:k-1}^\ell) },\)
      for $j = \range{1}{\Np}$. 
    \item Set $L \gets 0$
    \item \textbf{for $j=1$ to $\Np$}
      \begin{enumerate}
      \item Draw $X_k^i \sim \bar q_k(\cdot \mid X_{1:k-1}^j)$ and let $X_{1:k}^i = (X_{1:k-1}^j, X_{k}^i)$
        for $i = \range{L+1}{L+m_k^j}$.
      \item Set $L \gets L+m_k^j$.
      \end{enumerate}
    \end{enumerate}
  \end{enumerate}
\end{algorithm}

As mentioned above, at each iteration $k = \range{1}{n}$,
the method produces \emph{unweighted} samples $\{X_k^i\}_{i=1}^\Np$ approximating $\target_k$.
It also produces an unbiased estimate $\widehat Z_{\pi_k}$ of $Z_{\pi_k}$ \citep[Proposition~7.4.1]{DelMoral:2004}.
The algorithm is expressed in a slightly non-standard form; at iteration $k$ we loop over the ancestor particles, \ie the particles after resampling at
iteration $k-1$, and let each ancestor particle $j$ generate $m_k^j$ offsprings. (The variable $L$ is just for bookkeeping.)
This is done to clarify the connection with the \nsmc procedure below.
Furthermore, we have included a (completely superfluous) resampling step at iteration $k=1$, where the ``dummy variables''
$\{X_{1:0}^i\}_{i=1}^\Np$ are resampled according to the (all equal) weights $\{\Zprop{1}(X_{1:0}^i)\}_{i=1}^\Np = \{\Ztarget{1}\}_{i=1}^\Np$.
The analogue of this step is, however, used in the \nsmc algorithm, where the initial normalising constant
$Z_{\pi_1}$ is \emph{estimated}. We thus have to resample the corresponding initial particle systems accordingly.

\subsection{Fully Adapted Nested SMC Samplers}\label{sec:nsmc:fapf}
In analogue with Section~\ref{sec:nested}, assume now that we are not able to simulate exactly
from $\bar q_k$, nor compute $\Zprop{k}$.
Instead, we have access to a class $\Class$ which satisfies condition \asmpref{asmp:proper}.
The proposed \nsmc method is then given by Algorithm~\ref{alg:nsmc:nested-fapf}.

\begin{algorithm}
  \caption{Nested SMC (fully adapted)}
  \label{alg:nsmc:nested-fapf}
  \begin{enumerate}
  \item 
    Set $\widehat Z_{\pi_0} = 1$.
  \item \textbf{for $k=1$ to $n$}
    \begin{enumerate}
    \item Initialise $\Obj^j = \Class(q_k(\cdot \mid X_{1:k-1}^j), M)$ for $j = \range{1}{\Np}$.
    \item Set $\widehat Z_{q_k}^j = \Obj^j.\GetZ()$ for $j = \range{1}{\Np}$.
    \item Compute $\widehat Z_{\pi_k} = \widehat Z_{\pi_{k-1}}\times\left\{ \frac{1}{\Np}\sum_{j=1}^\Np \widehat Z_{q_k}^j \right\}.$
    \item Draw $m_k^{1:\Np}$ from a multinomial distribution with probabilities
      \( \frac{\widehat Z_{q_k}^j}{ \sum_{\ell=1}^\Np \widehat Z_{q_k}^\ell } \)
      for $j = \range{1}{\Np}$. 
    \item Set $L \gets 0$
    \item \textbf{for $j=1$ to $\Np$}
      \begin{enumerate}
      \item \label{item:sim-iii}Compute $X_k^i = \Obj^j.\Simulate()$ and let $X_{1:k}^i = (X_{1:k-1}^j, X_{k}^i)$
        for $i = \range{L+1}{L+m_k^j}$.
      \item \textbf{delete} $\Obj^j$.
      \item Set $L \gets L+m_k^j$.
      \end{enumerate}
    \end{enumerate}
  \end{enumerate}
\end{algorithm}

Algorithm~\ref{alg:nsmc:nested-fapf} can be seen as an \emph{exact approximation}
of the fully adapted \smc sampler in Algorithm~\ref{alg:nsmc:fapf}. (In Appendix~\ref{sec:supp:nsmc} we provide a formulation of
\nsmc with arbitrary proposals and resampling weights.)
We replace the exact computation of $\Zprop{k}$ and exact simulation from $\bar q_k$,
by the approximate procedures available through $\Class$.
Despite this approximation, however, Algorithm~\ref{alg:nsmc:nested-fapf} is a valid
\smc method. This is formalised by the following theorem.

\begin{theorem}
  \label{thm:clt}
  Assume that $\Class$ satisfies condition~\asmpref{asmp:proper}.
  Then, under certain regularity conditions on the function $f: \setX_k \mapsto \reals^d$
  and for an asymptotic variance $\Sigma_k^M(f)$, both specified in Appendix~\ref{sec:supp:clt}, we have
  \begin{align*}
   N^{1/2}\left( \frac{1}{N}\sum_{i=1}^N  f(X_{1:k}^i) - \target_k(f) \right) \convD \N(0, \Sigma_k^M(f)),
  \end{align*}
  where $\{X_{1:k}^i \}_{i=1}^M$ are generated by Algorithm~\ref{alg:nsmc:nested-fapf} and $\convD$ denotes convergence
  in distribution.
 \end{theorem}
\begin{proof}
  See Appendix~\ref{sec:supp:clt}.
\end{proof}

\begin{remark}
  The key point with Theorem~\ref{thm:clt} is that, under certain regularity conditions, the \nsmc
  method converges at rate $\sqrt{N}$ even for a fixed (and finite) value of the precision parameter
  $\Mp$. The asymptotic variance $\Sigma_k^M(f)$, however, will depend on the accuracy and
  properties of the approximative procedures of $\Class$. 
  We leave it as future work to establish more informative results, relating the asymptotic variance of \nsmc to that of the
  ideal, fully adapted \smc sampler.
\end{remark}

\subsection{Backward Simulation and Modularity of \nsmc}
As previously mentioned, the \nsmc procedure is modular in the sense
that we can make use of Algorithm~\ref{alg:nsmc:nested-fapf} also to define the class $\Class$.
Thus, we now view $\target_n$ as the \emph{proposal distribution} that we wish to approximately sample from using \nsmc.
Algorithm~\ref{alg:nsmc:nested-fapf} directly generates an estimate $\widehat Z_{\pi_n}$ of the normalising
constant of $\utarget_n$ (which indeed is unbiased, see Theorem~\ref{thm:proper}).
However, we also need to generate a sample $\widetilde X_{1:n}$ such that
$(\widetilde X_{1:n}, \widehat Z_{\pi_n})$ is properly weighted for~$\utarget_n$.

The simplest approach, akin to the Nested \is procedure described in Section~\ref{sec:nested:modular},
is to draw $B_n$ uniformly on $\crange{1}{\Np}$ and return $\widetilde X_{1:n} = X_{1:n}^{B_n}$.
This will indeed result in a valid definition of the $\Simulate$ procedure. However, this approach will suffer from the well known path degeneracy of \smc samplers. In particular, since we call $\Obj^j.\Simulate()$ multiple times in Step~\ref{item:sim-iii} of Algorithm~\ref{alg:nsmc:nested-fapf}, we risk to obtain (very) strongly correlated samples by this simple approach. %(if the 

It is possible to improve the performance of the above procedure by instead making use of a \emph{backward simulator}
\citep{GodsillDW:2004,LindstenS:2013} to simulate $\widetilde X_{1:n}$.
The backward simulator, given in Algorithm~\ref{alg:nsmc:bs}, is a type of smoothing algorithm;
it makes use of the particles generated by a forward pass of Algorithm~\ref{alg:nsmc:nested-fapf}
to simulate backward in ``time'' a trajectory $\widetilde X_{1:n}$ approximately
distributed according to $\target_n$.

\begin{algorithm}
  \caption{Backward simulator (fully adapted)} 
  \label{alg:nsmc:bs}
  \begin{enumerate}
  \item Draw $B_n$ uniformly on $\crange{1}{\Np}$.
  \item Set $\widetilde X_n = X_n^{B_n}$.
  \item \textbf{for $k=n-1$ to $1$}
    \begin{enumerate}
    \item Compute ${\displaystyle \widetilde W_k^j = \frac{ \utarget_n( (X_{1:k}^j, \widetilde X_{k+1:n}) ) }{ \utarget_k(X_{1:k}^j) }}$ for $j = \range{1}{\Np}$.
    \item Draw $B_k$ from a categorical distribution with probabilities
      ${\displaystyle \frac{ \widetilde W_k^j }{ \sum_{\ell=1}^\Np  \widetilde W_k^\ell }}$ for $j = \range{1}{\Np}$.
    \item Set $\widetilde X_{k:n} = (X_k^{B_k}, \widetilde X_{k+1:n})$.
    \end{enumerate}
  \end{enumerate}
\end{algorithm}

\begin{remark}
  Algorithm~\ref{alg:nsmc:bs} assumes unweighted particles and can thus be used in conjunction with the fully
  adapted \nsmc procedure of Algorithm~\ref{alg:nsmc:fapf}. If, however, the forward filter is not fully adapted
  the weights need to be accounted for in the backward simulation; see Appendix~\ref{sec:supp:proper}.
\end{remark}

The modularity of \nsmc is established by the following result.
\begin{definition}
  \label{def:class}
  Let $\Obj[p] = \Class(\utarget_n, N)$ be defined as follows:
  \vspace*{-3mm}
  \begin{enumerate}
  \setlength{\itemsep}{-1mm}
  \item The constructor executes Algorithm~\ref{alg:nsmc:nested-fapf} with target distribution $\utarget_n$
    and with $\Np$ particles, and $\Obj[p].\GetZ()$ returns the estimate of the normalising constant $\widehat Z_{\pi_n}$.
  \item $\Obj[p].\Simulate()$ executes Algorithm~\ref{alg:nsmc:bs} and returns $\widetilde X_{1:n}$.
  \end{enumerate}
\end{definition}

\begin{theorem}
  \label{thm:proper}
  The class $\Class$ defined as in Definition~\ref{def:class} satisfies condition~\asmpref{asmp:proper}.
\end{theorem}
\begin{proof}
  See Appendix~\ref{sec:supp:proper}.
\end{proof}
A direct, and important, consequence of Theorem~\ref{thm:proper} is that \nsmc can be used as a component of powerful
learning algorithms, such as the particle Markov chain Monte Carlo (\pmcmc) method \citep{andrieuDH2010particle}
and many of the other methods discussed in Section~\ref{sec:rw}.
Since standard \smc is a special case of \nsmc, Theorem~\ref{thm:proper}
implies proper weighting also of \smc.

%%% Local Variables:
%%% mode: latex
%%% TeX-master: "main.tex"
%%% End:

%% file: rw.tex
\section{Practicalities and Related Work} \label{sec:rw}

There has been much recent interest in using \smc within \smc in various ways.
The SMC$^2$ by \citet{chopinJP2013smc2} and the recent method by \citet{crisanM2013nested}
are sequential learning algorithms for state space models, where one
\smc sampler for the parameters is coupled with another \smc sampler for the latent states.
\citet{johansenWD2012exact} and \citet{chenSOL2011decentralized} address the state inference problem
by splitting the state variable into different components and run coupled \smc samplers for these components.
These methods differ substantially from \nsmc; they solve different problems and
the ``internal'' \smc sampler(s) is constructed in a different way (for approximate marginalisation instead
of for approximate simulation).
Another related method is the random weights \pf of \citet{fearnheadPRS2010random}, requiring exact samples from $\bar q$ and where the importance weights are estimated using a nested Monte Carlo algorithm.

The method most closely related to \nsmc is the space-time particle filter (\stpf) \citep{beskosCJKZ2014a},
which has been developed independently and in parallel with our work. The \stpf is also
designed for solving inference problems in high-dimensional models. It
can be seen as a island \pf \citep{vergeDDM2013on} implementation of the method presented by
\citet{naessethls2014sequential}. Specifically, for a spatio-temporal models they run an island \pf
over both spatial and temporal dimensions. However, the \stpf does not generate an approximation of the fully adapted \smc sampler. 

Another key distinction between \nsmc and \stpf is that in the latter each particle in the ``outer'' \smc sampler
comprises a complete particle system from the ``inner'' \smc sampler. For \nsmc, on the other hand,
the particles will simply correspond to different hypotheses about the latent variables (as in standard \smc),
regardless of how many samplers that are nested. This is a key feature of \nsmc, since it implies that it is easily distributed over the particles. The main computational effort of Algorithm~\ref{alg:nsmc:nested-fapf}
is the construction of $\{\Obj^j\}_{j=1}^N$ and the calls to the $\Simulate$ procedure,
which can be done independently for each particle. This leads to improved memory efficiency and lower communication costs.
Furthermore, we have found (see Section~\ref{sec:expts}) that \nsmc can outperform \stpf even when
run on a single machine with matched computational costs.

Another strength of \nsmc methods are their relative ease of implementation, which we show in Section~\ref{sec:drought}. We use the framework to sample from what is essentially a cubic grid Markov random field (\mrf) model just by implementing three
nested samplers, each with a target distribution defined on a simple chain.

There are also other \smc-based methods designed for high-dimensional problems, \eg,
the block \pf studied by \citet{rebeschiniH2015can}, the
location particle smoother by \citet{briggsDM2013data} and the \pf-based methods reviewed in \citet{djuric2013particle}. 
However, these methods are all inconsistent, as they are based on various approximations that
result in systematic errors.

The previously mentioned \pmcmc \citep{andrieuDH2010particle} is a related method, where \smc is used as a component of an \mcmc algorithm. We make use of a very similar extended space approach to motivate the validity of our algorithm.
Note that our proposed algorithm can be used as a component in
\pmcmc and most of the other algorithms mentioned above, which further increases the scope of models it can handle.

%%% Local Variables:
%%% mode: latex
%%% TeX-master: "main.tex"
%%% End:

%% file: expts.tex
\section{Experimental Results}\label{sec:expts}
We illustrate \nsmc on three high-dimensional examples, both with real and synthetic data. We compare \nsmc with standard (bootstrap) \pf and the \stpf of \citet{beskosCJKZ2014a} with equal computational budgets on a single machine (\ie, neglecting the fact that \nsmc is more easily distributed). These methods are, to the best of our knowledge, the only other available \emph{consistent} online methods for full Bayesian inference in general sequential models. 
For more detailed explanations of the models and additional results, see Appendix~\ref{sec:supp:expts}\footnote{Code available at \url{https://github.com/can-cs/nestedsmc}}.

\setlength{\tabcolsep}{0.5pt}
\renewcommand{\arraystretch}{0.5}
\begin{figure*}[tb]
    \centering
	\begin{tabular}{m{.01\textwidth} m{.01\textwidth} m{.3\textwidth} m{.3\textwidth} m{.3\textwidth} m{.000005\textwidth} }
	 & & \centering \small $d=50$ & \centering \small $d=100$ & \centering \small $d=200$ & \\
	 \rotatebox{90}{ESS} & & \includegraphics[width=.3\textwidth]{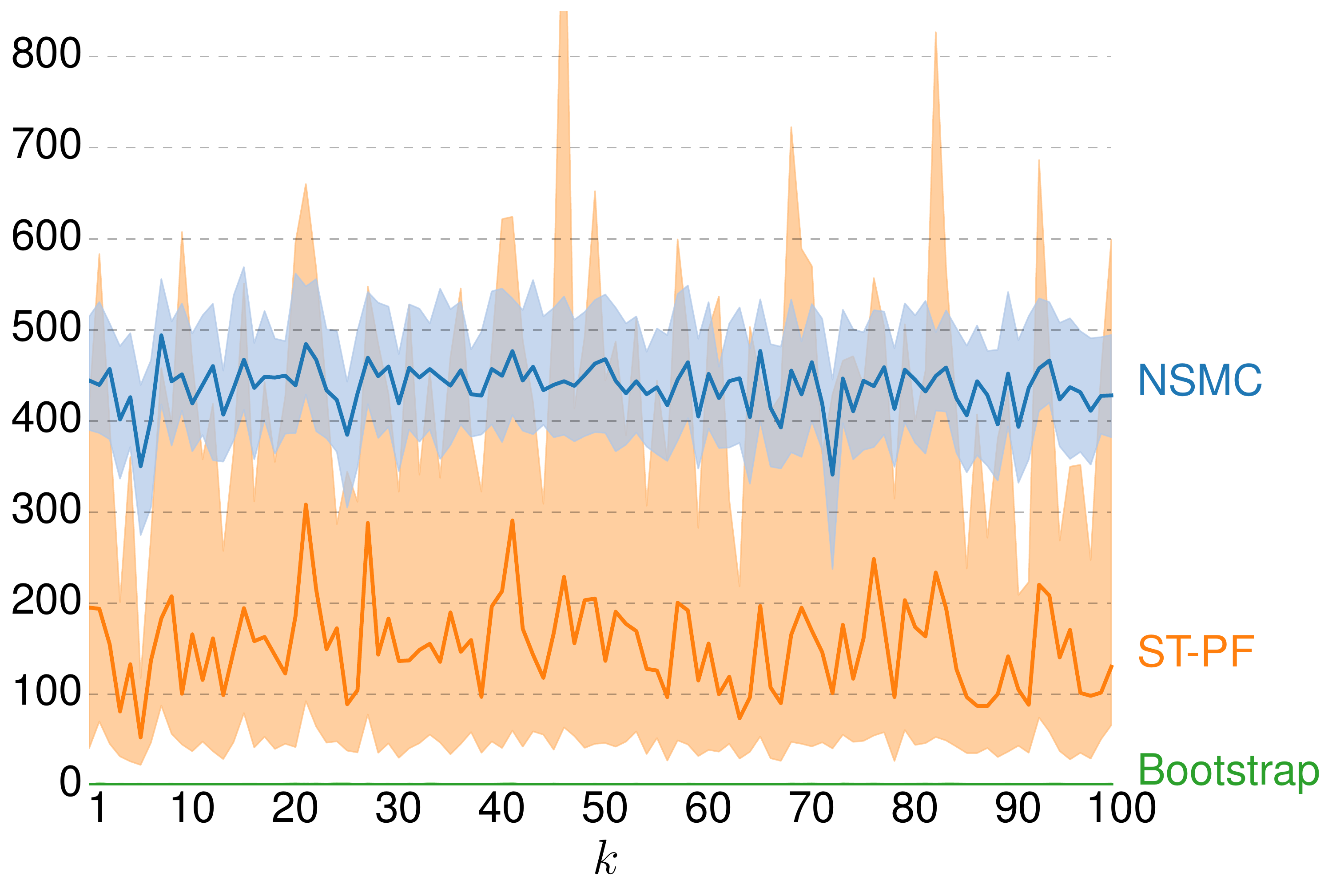} & \includegraphics[width=.3\textwidth]{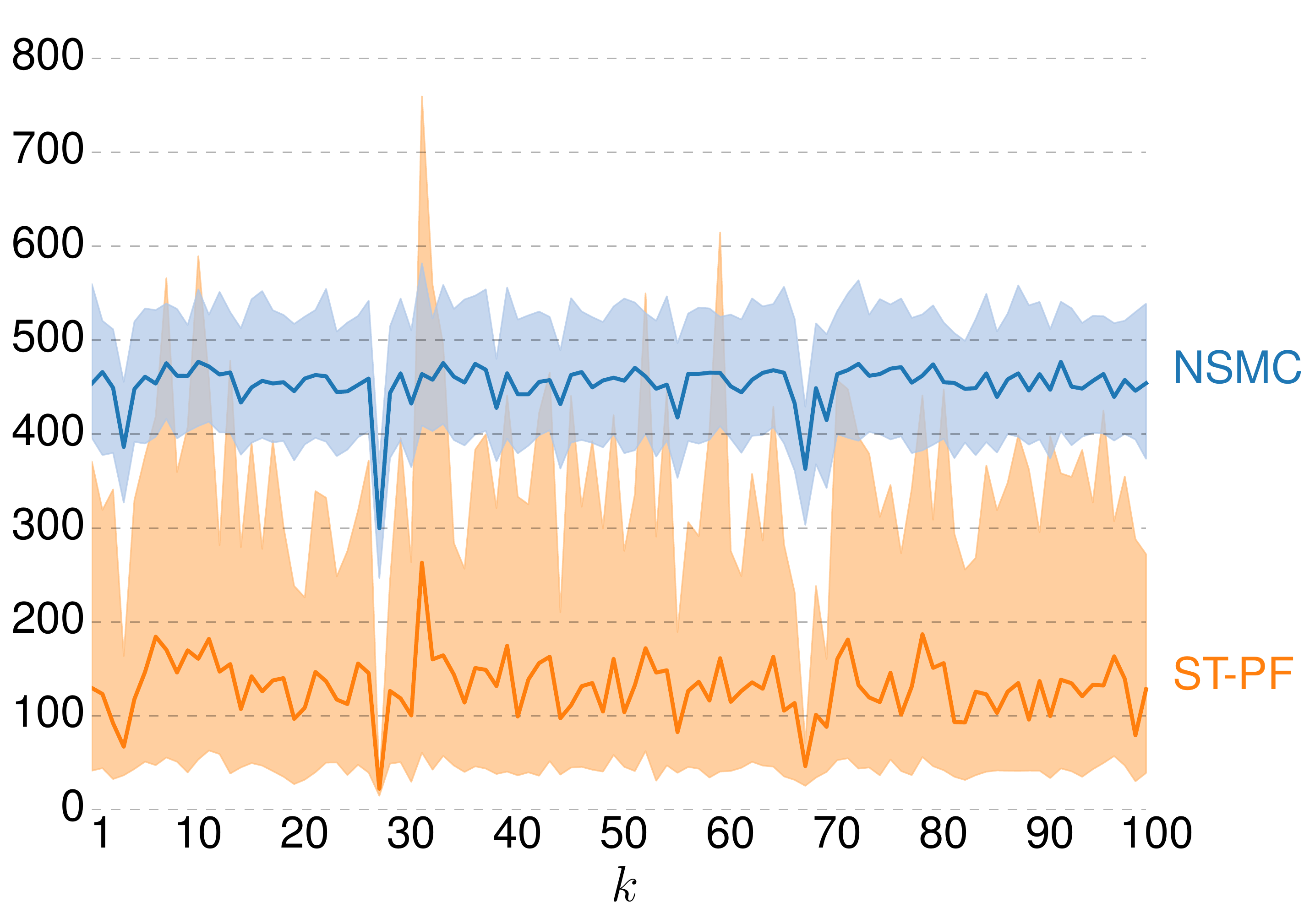} & \includegraphics[width=.3\textwidth]{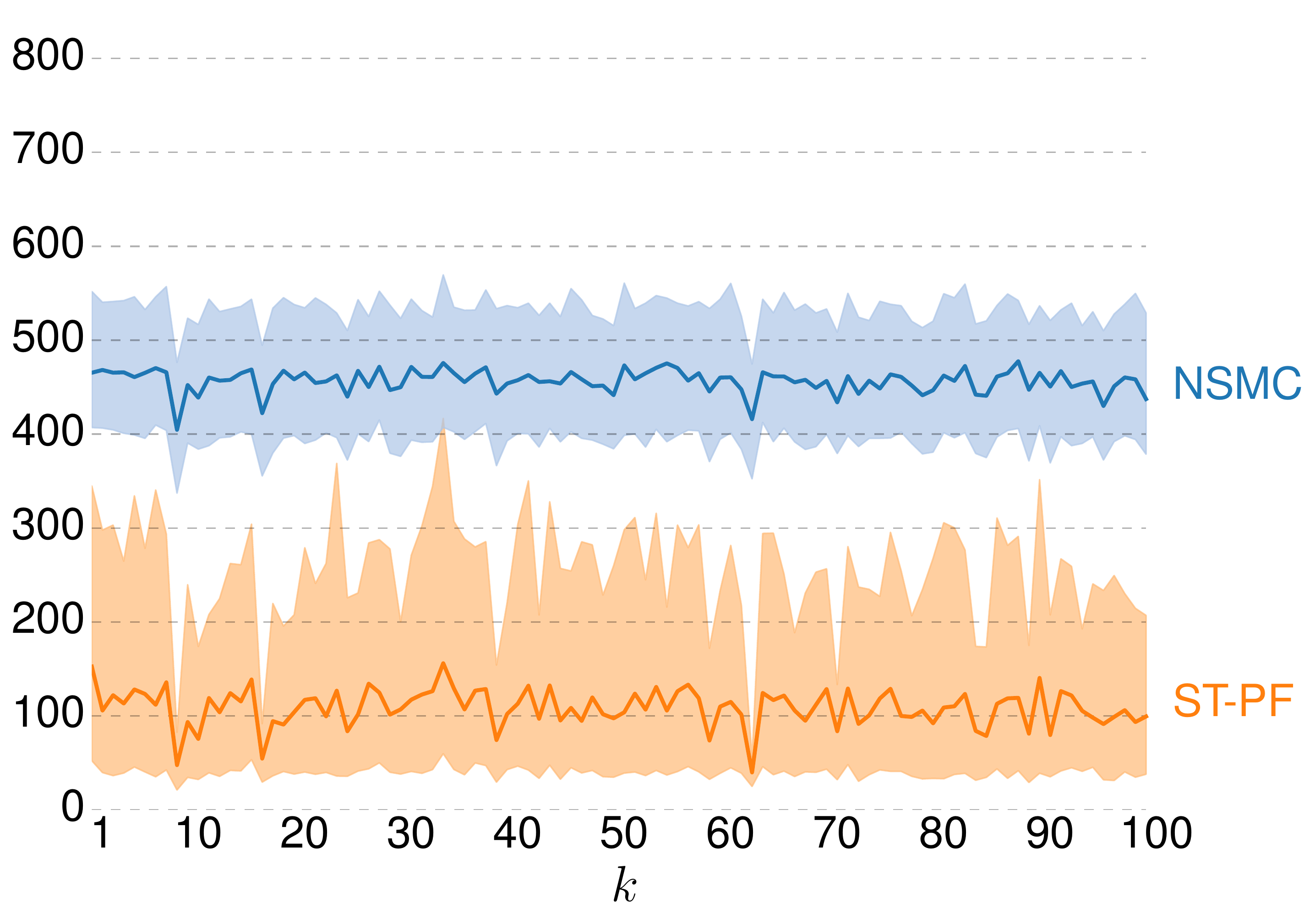} & \\
	 \rotatebox{90}{ERS} & & \includegraphics[width=.3\textwidth]{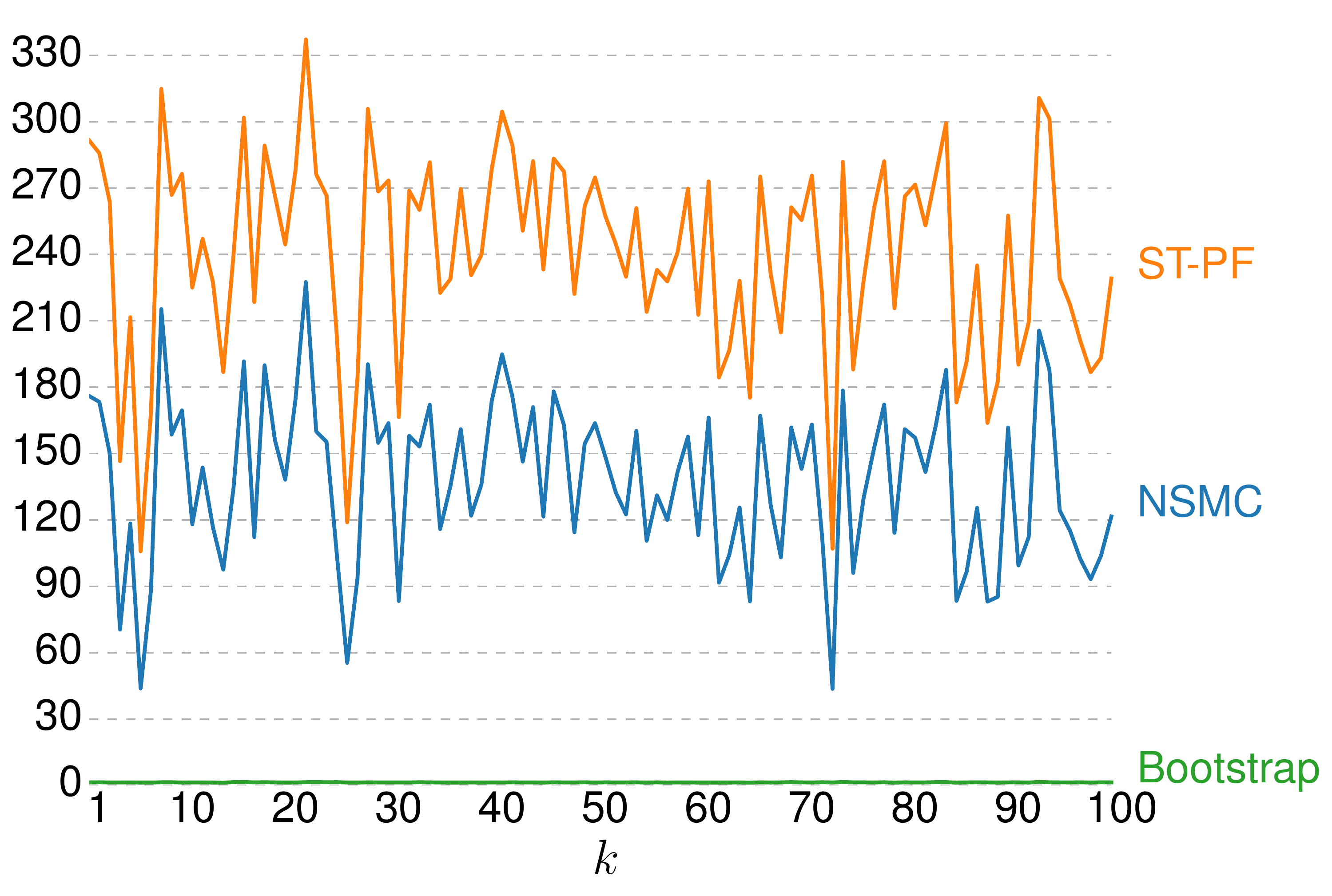} & \includegraphics[width=.3\textwidth]{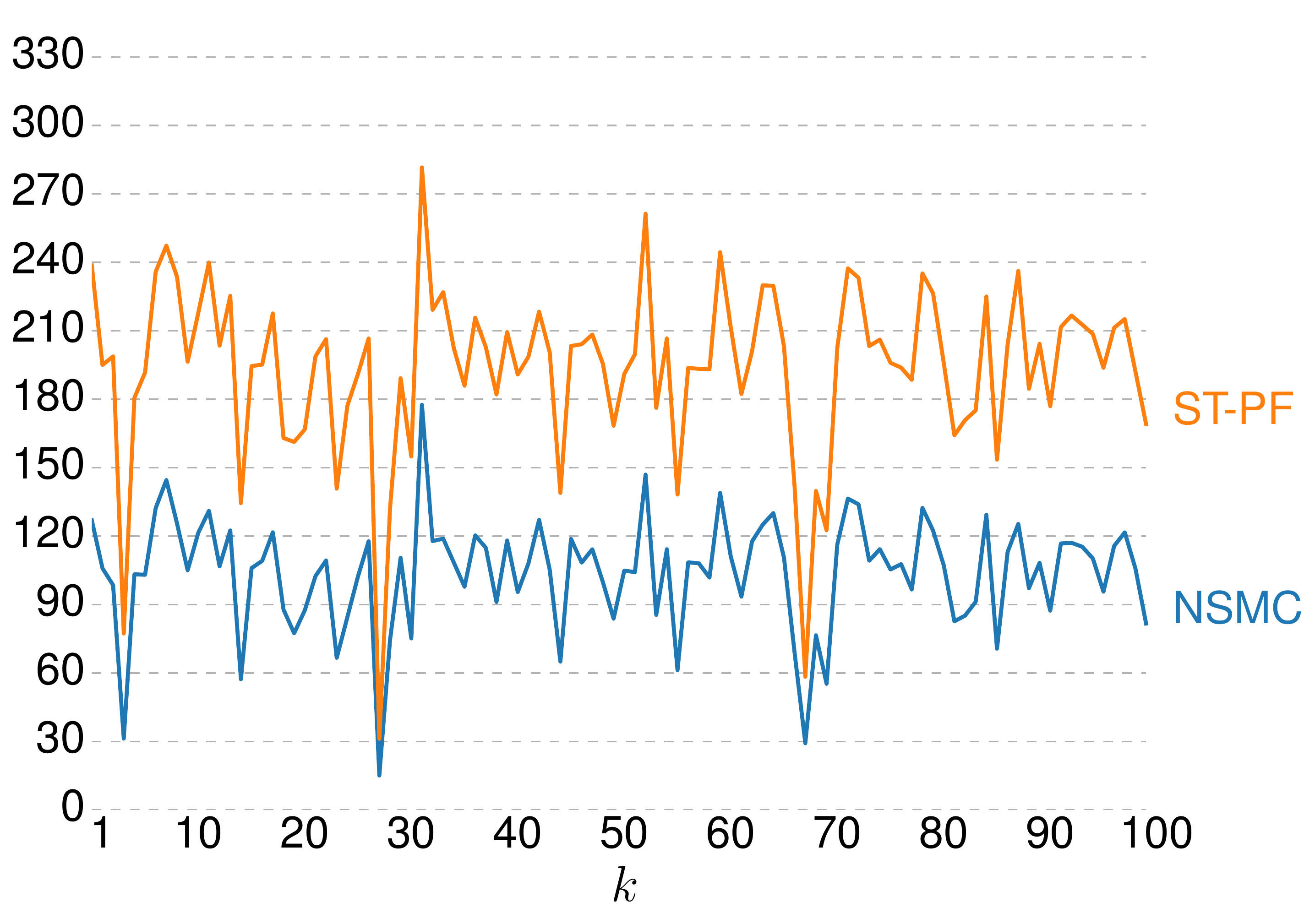} & \includegraphics[width=.3\textwidth]{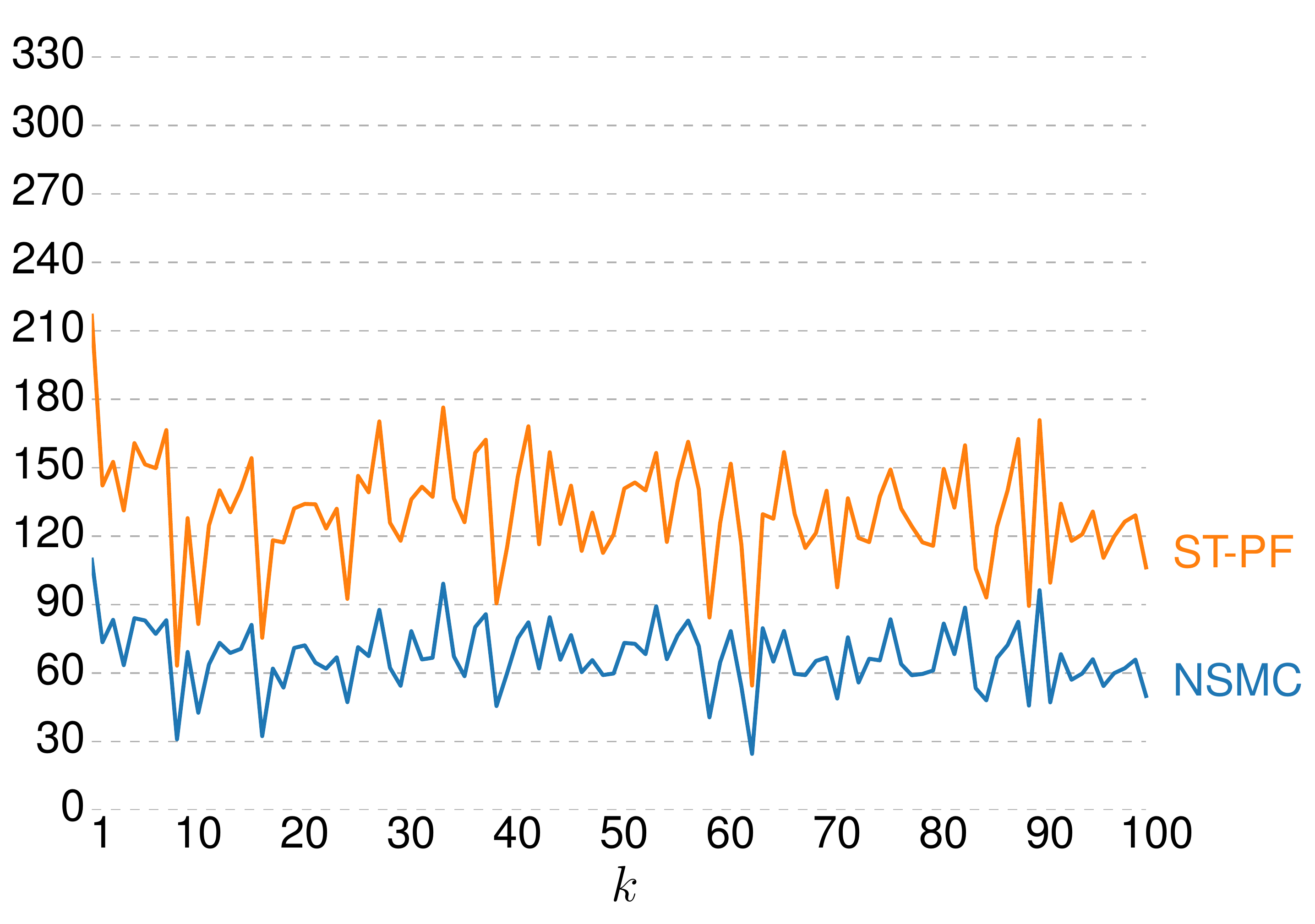} &
	\end{tabular}
	\caption{\emph{Top:} Median (over dimension) \ess~\eqref{eq:est:ess} and $15$--$85$\% percentiles (shaded region). \emph{Bottom:} The \ers~\eqref{eq:est:ers} based on the resampling weights in the (outermost) particle filter. The results are based on $100$ independent runs for the Gaussian \mrf with dimension $d$.}\label{fig:lgss}
\end{figure*}

\subsection{Gaussian State Space Model}
We start by considering a high-dimensional Gaussian state space model,
where we have access to the true solution through belief propagation. %from the Kalman filter \citep{Kalman:1960}.
The latent variables and measurements $\{ X_{1:k}, Y_{1:k} \}$,
with $\{X_k, Y_k\} = \left\{ X_{k,l}, Y_{k,l} \right\}_{l=1}^d$, are modeled by
a $d \times k$ lattice Gaussian \mrf. The true data is simulated from a nearly identical state space model (see Appendix~\ref{sec:supp:expts:lgss}). %, which can be identified with a linear Gaussian state space model (see Appendix~\ref{sec:supp:expts:lgss}).
We run a 2-level \nsmc sampler. The outer level is fully adapted, \ie the
proposal distribution is $q_k = p(x_k \mid x_{k-1}, y_k)$, which thus constitute the target distribution
for the inner level. 
To generate properly weighted samples from $q_k$, we use a bootstrap \pf
operating on the $d$ components of the vector $x_k$.
Note that we only use bootstrap proposals where the actual sampling takes place, and that
the conditional distribution $p(x_k \mid x_{k-1}, y_k)$ is not explicitly used.

We simulate data from this model for $k=1,\ldots,100$ for different values of $d = \text{dim}(x_k) \in \{ 50, 100, 200\}$.
The exact filtering marginals are computed using belief propagation.%the Kalman filter. 
We compare with both the \stpf and standard (bootstrap) \pf.

The results are evaluated based on the effective sample size (\ess, see \eg \citet{fearnheadWT2010a})
defined as,
\begin{align}
  \label{eq:est:ess}
  \text{ESS}(x_{k,l}) = \left( \E\left[ {\textstyle \frac{(\widehat{x}_{k,l} - \mu_{k,l})^2}{\sigma_{k,l}^{2}} } \right]\right)^{-1},
\end{align}
where $\widehat{x}_{k,l}$ denote the mean estimates and $\mu_{k,l}$ and $\sigma_{k,l}^{2}$ denote the true mean and variance of $x_{k,l}\mid y_{1:k}$ obtained from belief propagation. %the Kalman filter. 
The expectation in~\eqref{eq:est:ess} is approximated by averaging over $100$ independent runs of the involved algorithms.
The \ess reflects the estimator accuracy, obvious by the definition which is tightly related to the
mean-squared-error. Intuitively the \ess corresponds to the equivalent number of \iid samples needed for the same accuracy. 

We also consider the effective resample size (\ers, \citet{KongLW:1994}), which is based on the resampling weights at the top levels in the respective \smc algorithms, 
\begin{align}
  \textstyle 
  \text{ERS} = \frac{\left(\sum_{i=1}^N \widehat Z_{q_k}^i\right)^2}{\sum_{i=1}^N \left(\widehat Z_{q_k}^i\right)^2}.\label{eq:est:ers}
\end{align}
The \ers is an estimate of the effective number of unique particles (or particle systems in the case of \stpf) available at each resampling step.

We use $N = 500$ and $M = 2\cdot d$ for \nsmc and match the computational time for \stpf and bootstrap \pf.
We report the results in Figure~\ref{fig:lgss}. The bootstrap \pf is omitted from $d=100$, $200$ due to its poor performance already for $d=50$ (which is to be expected).
Each dimension $l=1,\ldots,d$ provides us with a value of the \ess, so we present the median (lines)
and $15$--$85$\% percentiles (shaded regions) in the first row of Figure~\ref{fig:lgss}. 
The \ers is displayed in the second row of Figure~\ref{fig:lgss}. Note that \ess gives a better reflection of estimation accuracy than \ers.

We have conducted additional experiments with different model parameters and different choices for
$N$ and $M$ (some additional results are given in Appendix~\ref{sec:supp:expts:lgss}). Overall the results seem to be in agreement with the ones presented here, however \stpf seems to be more robust to the trade-off between $N$ and $M$. A rule-of-thumb for \nsmc is to generally try to keep $N$ as high as possible, while still maintaining a reasonably large \ers.

\subsection{Non-Gaussian State Space Model}
Next, we consider an example with a non-Gaussian \ssm, borrowed from \citet{beskosCJKZ2014a} where the full details of the model are given.
The transition probability $p(x_k \mid x_{k-1})$ is a localised Gaussian mixture and the measurement probability $p(y_k \mid x_k)$ is t-distributed. The model dimension is $d=1\thinspace024$.
  \citet{beskosCJKZ2014a} report improvements for \stpf over both the bootstrap \pf
and the block \pf by \citet{rebeschiniH2015can}. 
 We use $N = M = 100$
for both \stpf and \nsmc (the special structure of this model implies that there is no significant computational overhead from 
\begin{figure}[h]
\begin{center}
    \includegraphics[width=0.5\columnwidth]{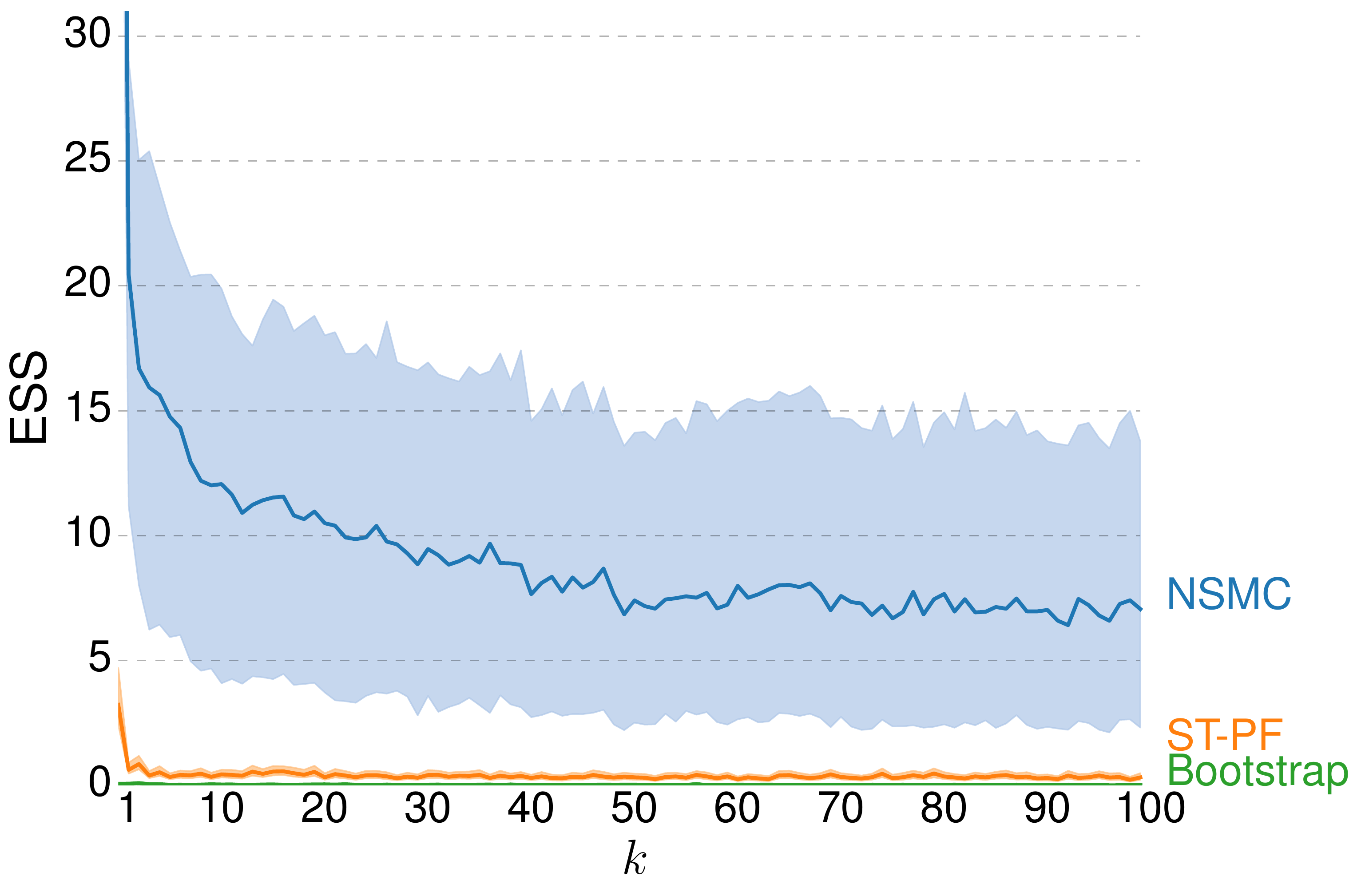}
  \end{center}
\caption{Median \ess with $15-85\%$ percentiles (shaded region) for the non-Gaussian \ssm.}
\label{fig:nless}
\end{figure}
running backward sampling) and the bootstrap \pf is given $N=\thsnd{10}$. In Figure~\ref{fig:nless} we report the \ess~\eqref{eq:est:ess}, estimated according to \citet{CarpenterCF:1999}. The \ess for the bootstrap \pf is close to $0$, for \stpf around 1--2,
and for \nsmc slightly higher at 7--8. However, we note that all methods perform quite poorly on this model,
and to obtain satisfactory results it would be necessary to use more particles.

\subsection{Spatio-Temporal Model -- Drought Detection}\label{sec:drought}
In this final example we study the problem of detecting droughts based on measured precipitation data \citep{droughtdata} for different locations on earth. We look at the situation in North America during the years $1901$--$1950$ and the Sahel region in Africa during the years $1950$--$2000$. These spatial regions and time frames were chosen since they include two of the most devastating droughts during the last century, the so-called Dust Bowl in the US during the 1930s \citep{SchubertSPKB:2004} and the decades long drought in the Sahel region in Africa starting in the 1960s \citep{FoleyCSW:2003,HoerlingHEP:2006}.
\begin{figure}[h]
  \centering
  \resizebox{0.4\columnwidth}{!} {
  \input{droughtNSMC.tex}
    }
  \caption{Illustration of the three-level \nsmc.}
  \label{fig:NSMCstructure4Draught}
\end{figure}
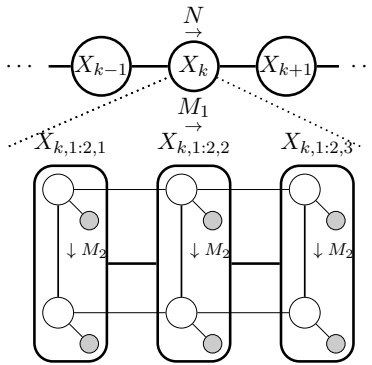
We consider the spatio-temporal model defined by \citet{fuBLS2012drought} and compare with the results therein. Each location in a region is modelled to be in either a \emph{normal} state $0$ or in an \emph{abnormal} state $1$ (drought). Measurements are given by precipitation (in millimeters) for each location and year. 
At every time instance $k$ our latent structure is described by a rectangular $2$D grid $X_k = \{X_{k,i,j}\}_{i=1,j=1}^{I,J}$; in essence this is the model showcased in Figure~\ref{fig:droughtmodel}. \citet{fuBLS2012drought} considers the problem of finding the maximum aposteriori configuration, using a linear programming relaxation. We will instead compute an approximation of the full posterior filtering distribution $\target_k(x_k) = p(x_k \mid y_{1:k})$.
\begin{figure*}[tb]
\centering
	\begin{tabular}{m{.16\textwidth} m{.16\textwidth} m{.16\textwidth} m{.16\textwidth} m{.16\textwidth} m{.16\textwidth}}
	\multicolumn{3}{c}{\includegraphics[width=.45\textwidth]{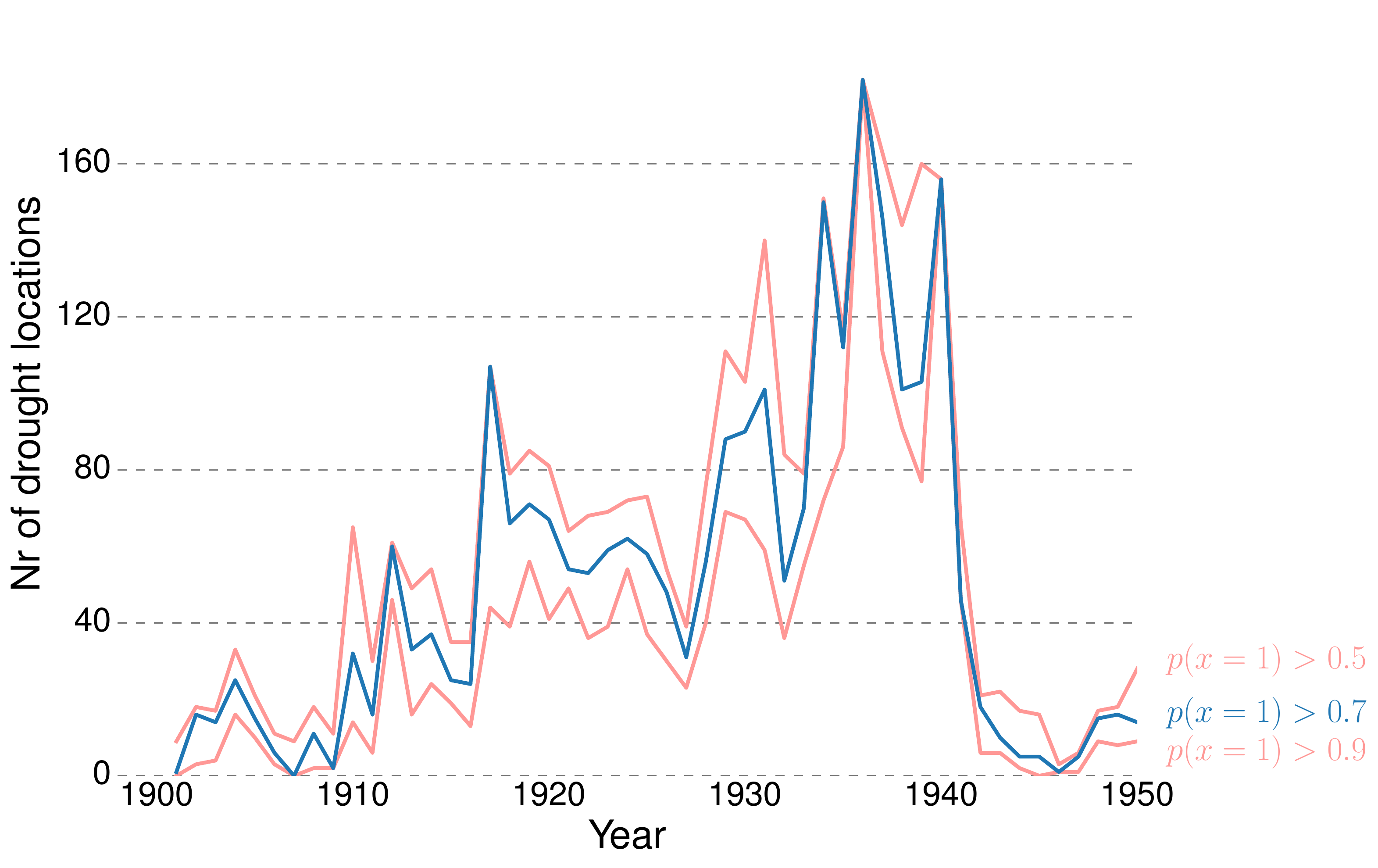}} & \multicolumn{3}{c}{\includegraphics[width=.45\textwidth]{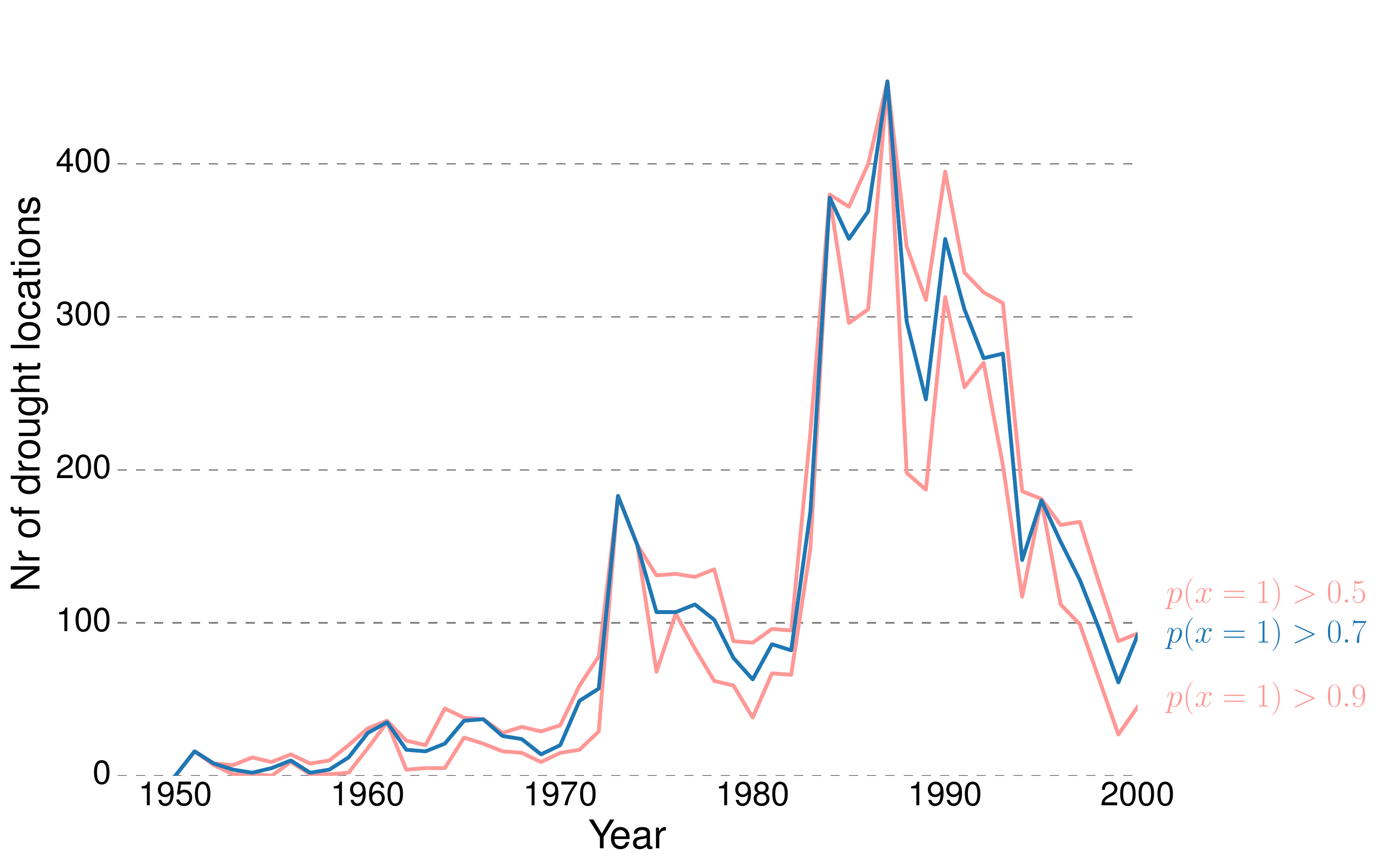}}\\
	\multicolumn{3}{c}{\small North America region} & \multicolumn{3}{c}{\small Sahel region} \\
	\multicolumn{2}{c}{\includegraphics[width=.3\textwidth]{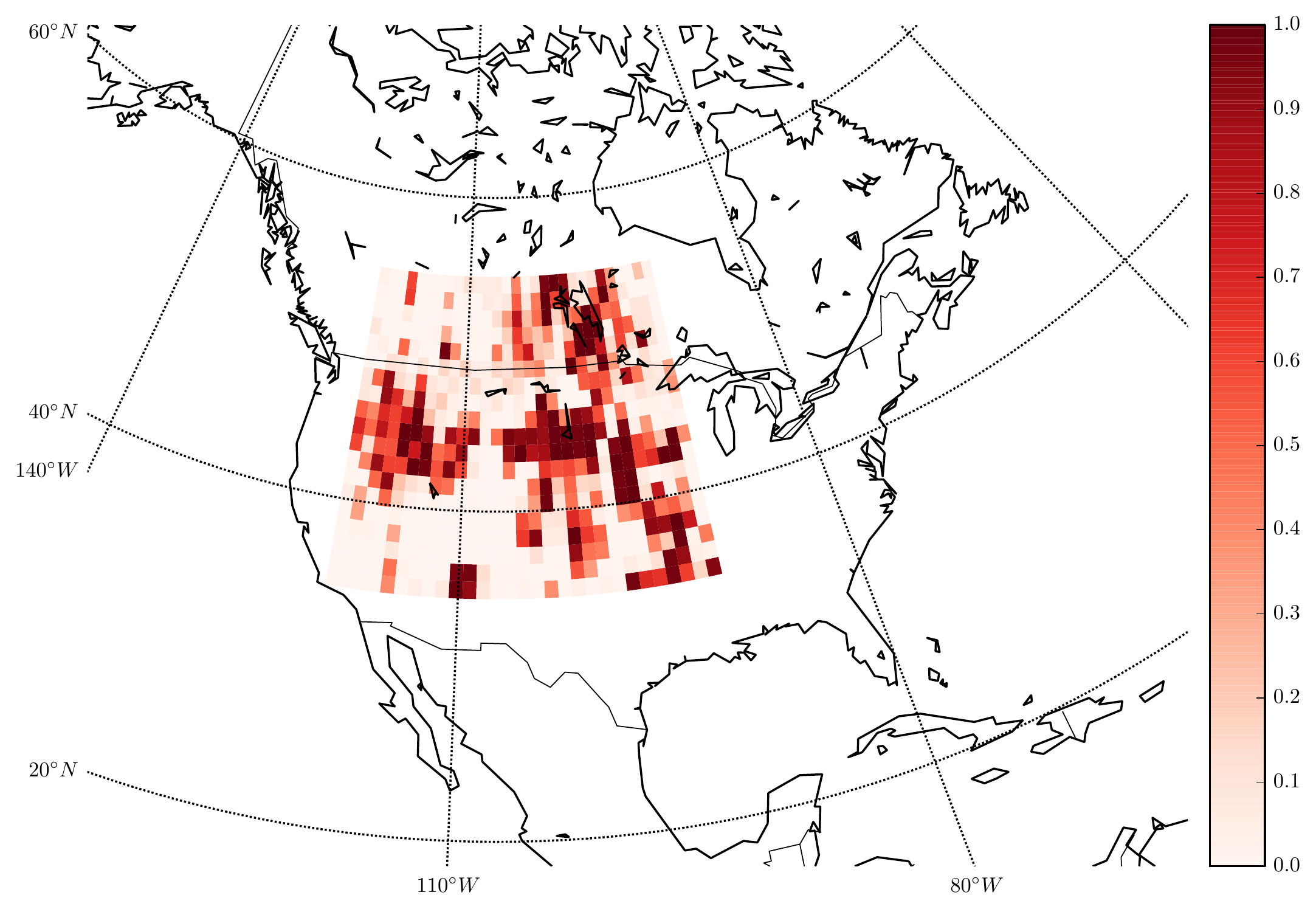}} & \multicolumn{2}{c}{\includegraphics[width=.3\textwidth]{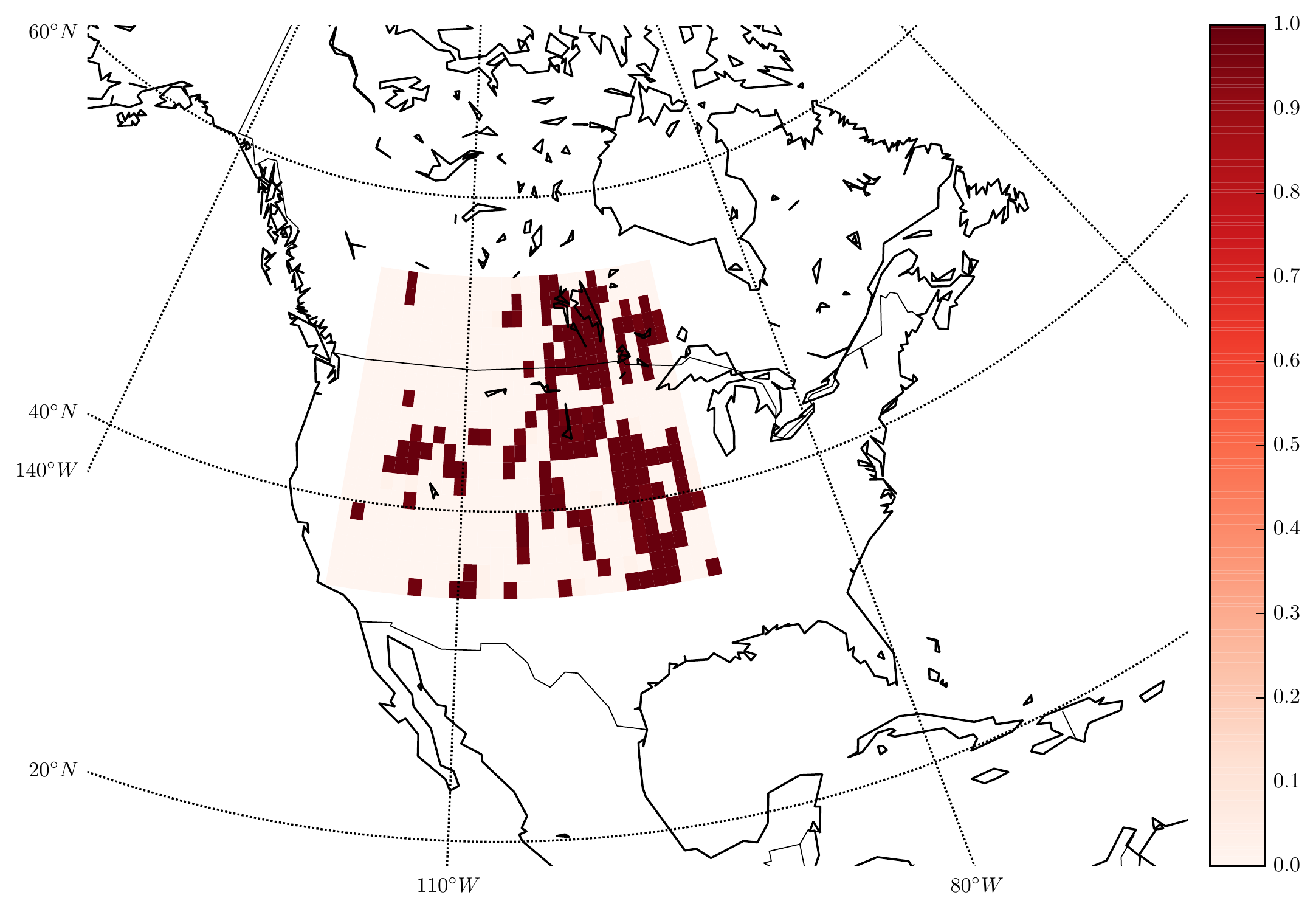}} & \multicolumn{2}{c}{\includegraphics[width=.3\textwidth]{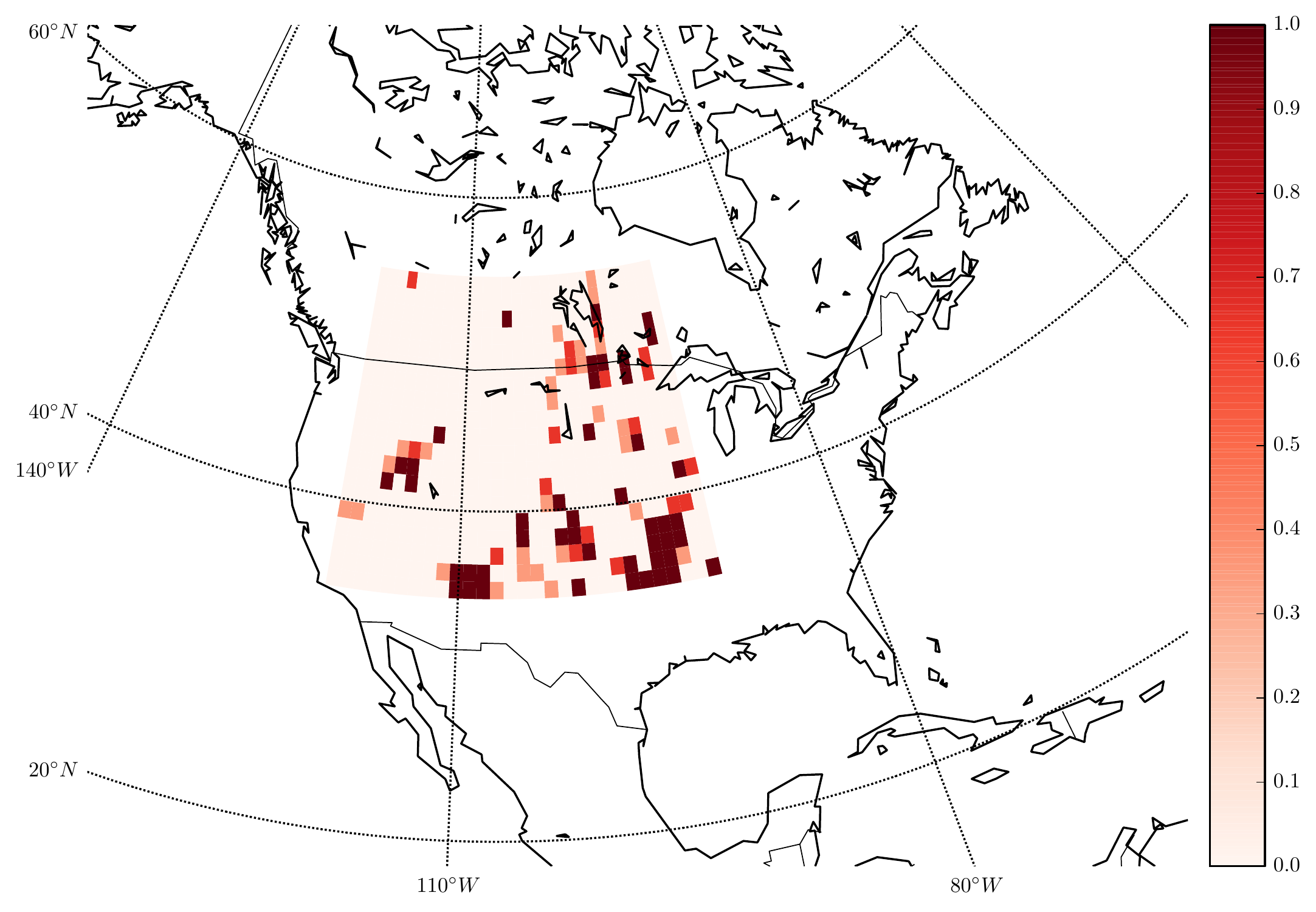}}\\
	\multicolumn{2}{c}{\small North America $1939$} & \multicolumn{2}{c}{\small North America $1940$} & \multicolumn{2}{c}{\small North America $1941$}
	\end{tabular}
	\caption{\emph{Top:} Number of locations with estimated $p(x=1) > \{0.5,0.7,0.9\}$ for the two regions. \emph{Bottom:} Estimate of $p(x_{t,i}=1)$ for all sites over a span of 3 years. All results for $\Np = 100, \Np_1 = \{30,40\}, \Np_2 = 20$.}\label{fig:drought}
\end{figure*}
The rectangular structure is used to instantiate an \nsmc method that on the first level targets the full posterior filtering distribution. To sample from $X_k$ we run, on the second level, an \nsmc procedure that operates on the ``columns'' $X_{k,1:I,j}$, $j = \range{1}{J}$. Finally, to sample each column $X_{k,1:I,j}$ we run a third level of \smc, that operates
on the individual components $X_{k,i,j}$, $i = \range{1}{I}$, using a bootstrap proposal.
The structure of our \nsmc method applied to this particular problem is illustrated in Figure~\ref{fig:NSMCstructure4Draught}.

Figure~\ref{fig:drought} gives the results on the parts of North America 
that we consider. 
The first row shows the number of locations where the estimate of 
 $p(x_{k,i,j}=1)$ exceeds $\{0.5, 0.7, 0.9\}$, for both regions. These results seems to be in agreement with \citet[Figures 3, 6]{fuBLS2012drought}. 
However, we also receive an approximation of the full posterior and can visualise uncertainty in our estimates,
as illustrated by the three different levels of posterior probability for drought.
In general, we obtain a rich sample diversity from the posterior distribution.
However, for some problematic years the sampler degenerates, with the result
that the three credibility levels all coincide. This is also visible in the second row of Figure~\ref{fig:drought},
where we show the posterior estimates $p(x_{k,i,j} \mid y_{1:k})$ for the years 1939--1941, overlayed on the regions of interest. For year 1940 the sampler degenerates and only reports 0-1 probabilities for all sites.
Naturally, one way to improve the estimates is to run the sampler with a larger number of particles, which has been
kept very low in this proof-of-concept.

%%% Local Variables:
%%% mode: latex
%%% TeX-master: "main.tex"
%%% End:

%% file: droughtNSMC.tex
\tikzstyle{edge} = [dotted, thick]
\tikzstyle{edge2} = [-,very thick]
\tikzstyle{edge3} = [-]
\tikzstyle{arrw} = [very thick,shorten <=2pt,shorten >=2pt]
\tikzstyle{var} = [draw,circle,inner sep=0,minimum width=0.5cm]
\tikzstyle{bigvar} = [draw,circle,inner sep=0,minimum width=0.8cm,very thick]
\tikzstyle{obs} = [draw,circle,inner sep=0,minimum width=0.3cm, fill=black!20]
  \begin{tikzpicture}[>=stealth,node distance=0.3cm]
    \begin{scope}
	  % Draw chain
	  \node at (-0.6,4) (x04) [circle] {$\cdots$};
	  \node at (0.7,4) (x14) [bigvar] {$X_{k-1}$};
	  \node at (2.2,4.7) {$\displaystyle \underset{\rightarrow}{N}$};
	  \node at (2.2,4) (x24) [bigvar] {$X_{k}$};
	  \node at (3.7,4) (x34) [bigvar] {$X_{k+1}$};
      \node at (5,4) (x44) [circle] {$\cdots$};
      \draw[edge2] (x04) -- (x14);
      \draw[edge2] (x14) -- (x24);
      \draw[edge2] (x24) -- (x34);
      \draw[edge2] (x34) -- (x44);
      \draw[dotted, thick] (-0.8,2.7) -- (x24);
      \draw[edge] (-0.8,2.7) -- (x24);
      \draw[edge] (4.9,2.7) -- (x24);
      
      % Draw x-nodes and observations
      \foreach \x in {0,2,4} {
        \foreach \y in {0,2} {
          \node at (\x,\y) (x\x\y) [var] {};
        }
      }
      \foreach \x in {0,2,4} {
        \foreach \y in {0,2} {
          \node (y\x\y) [obs,below right=of x\x\y] {};
        }
      }
      % Draw obs-var edge
      \foreach \x in {0,2,4} {
        \foreach \y in {0,2} {
          \draw[edge3] (x\x\y) -- (y\x\y);
        }
      }

      % Draw horizontal edges
      \foreach \x in {0,2} {
        \pgfmathtruncatemacro\xend{\x+2}
        \foreach \y in {0,2} {
          \draw[edge3] (x\x\y) -- (x\xend\y);
        }
      }
      \node at (2.2,3.2) {$\displaystyle \underset{\rightarrow}{M_1}$};
      % Draw vertical edges
      \foreach \x in {0,2,4} {
        \foreach \y in {0} {
          \pgfmathtruncatemacro\yend{\y+2}
          \path[draw, -, thick] (x\x\y) -- node[auto, swap] {\scriptsize $\downarrow M_2$}  (x\x\yend);
        }
      }

      \node[draw,very thick,rectangle,rounded corners=3mm,minimum width=1.1cm,fit=(x00) (y00) (x02) (y02),label=above:{$X_{k,1:2,1}$}] (x0){};
      \node[draw,very thick,rectangle,rounded corners=3mm,minimum width=1.1cm,fit=(x20) (y20) (x22) (y22),label=above:{$X_{k,1:2,2}$}] (x2){};
      \node[draw,very thick,rectangle,rounded corners=3mm,minimum width=1.1cm,fit=(x40) (y40) (x42) (y42),label=above:{$X_{k,1:2,3}$}] (x4){};
      
      %Draw fat horizontal edges
      \foreach \x in {0,2} {
        \pgfmathtruncatemacro\xend{\x+2}
          \path [draw, -, very thick] (x\x) -- (x\xend);
      }
    \end{scope}
     %\draw [blue] (current bounding box.south west) rectangle (current bounding box.north east);
  \end{tikzpicture}

%% file: suppMain.tex
%\externaldocument{main}

In this appendix we start out in Section~\ref{sec:supp:nsmc} by providing a more
general formulation of the \nsmc method and proofs of the central limit and proper weighting
theorems of the main manuscript. We also detail (Section~\ref{sec:supp:nis}) a straightforward
extension of nested \is to a sequential version. We show that a special case of this nested
sequential \is turns out to be more or less equivalent to the importance sampling squared algorithm
by \citet{tranSPK2013importance}. This relationship serves as evidence that illustrates that the
\nsmc framework being more widely applicable than the scope of problems considered in this
article. Finally, in Section~\ref{sec:supp:expts} we give more details and results on the
experiments considered in the main manuscript.

% Supplementary Nested SMC (APF)
\input{suppNSMC}

% Supplementary Nested IS
\input{suppNIS}

% Supplementary experiments
\input{suppExpts}

%% file: suppNSMC.tex
\subsection{Nested Sequential Monte Carlo}\label{sec:supp:nsmc}
We start by presenting a general formulation of a nested auxiliary \smc sampler in
Algorithm~\ref{alg:app:nested-smc}.  In this formulation, $q_k(x_k \mid x_{1:k-1})$ is an arbitrary
(unnormalised) proposal, normalised by 
\[
Z_{q_k}(x_{1:k-1}) = \int q_k(x_k \mid x_{1:k-1}) \myd x_k.
\]
Furthermore, the resampling weights are obtain by multiplying the importance weights with the
arbitrary \emph{adjustment multipliers} $\nu_{k-1}(x_{1:k-1}, Z_{q_k})$, which may depend on both the
state sequence $x_{1:k-1}$ and the normalising constant (estimate). The fully adapted \nsmc sampler
(Algorithm~\ref{alg:nsmc:nested-fapf} in the main document) is obtained as a special case if we
choose
\begin{align*}
  q_k (x_k \mid x_{1:k-1}) = \frac{\utarget_k(x_{1:k}) }{ \utarget_{k-1}(x_{1:k-1}) }
\end{align*}
and $\nu_{k-1}(x_{1:k-1}, Z_{q_k}) = Z_{q_k}$, in which case the importance weights are indeed given by $W_k^i \equiv 1$.

\begin{algorithm}[tb]
  \caption{Nested SMC (auxiliary SMC formulation)}
  \label{alg:app:nested-smc}
  \begin{enumerate}
  \item Set $\{ X_0^i \}_{i=1}^\Np$ to arbitrary dummy variables.
    Set $W_0^i = 1$ for $i=\range{1}{\Np}$.
    Set $\widehat Z_{\pi_0} = 1$.
  \item \textbf{for $k=1$ to $n$}
    \begin{enumerate}
    \item\label{step:app:nsmc:init-q} Initialise $\Obj^j = \Class(q_k(\cdot \mid X_{1:k-1}^j), M)$ for $j = \range{1}{\Np}$.
    \item\label{step:app:nsmc:get-Z} Compute $\widehat Z_{q_k}^j = \Obj^j.\GetZ()$ for $j = \range{1}{\Np}$.
    \item Compute $\widehat \nu_{k-1}^j = \nu_{k-1}(X_{1:k-1}^j, \widehat Z_{q_k}^j)$  for $j = \range{1}{\Np}$.
    \item\label{step:app:nsmc:resample} Draw $m_k^{1:\Np}$ from a multinomial distribution with probabilities
      \(\displaystyle \frac{\widehat\nu_{k-1}^j W_{k-1}^j}{ \sum_{\ell=1}^\Np \widehat\nu_{k-1}^\ell W_{k-1}^\ell} \)
      for $j = \range{1}{\Np}$. 
    \item Set $L \gets 0$
    \item \textbf{for $j=1$ to $\Np$}
      \begin{enumerate}
      \item\label{step:app:nsmcs:Prop} Compute $X_k^i = \Obj^j.\Simulate()$ and let $X_{1:k}^i = (X_{1:k-1}^j, X_{k}^i)$
        for $i = \range{L+1}{L+m_k^j}$.
      \item\label{step:app:nsmcs:Wght} Compute \({\displaystyle W_k^i = \frac{ \pi_k( X_{1:k}^i)  }{ \pi_{k-1}(X_{1:k-1}^j) } 
        \frac{ \widehat Z_{q_k}^j }{ \widehat\nu_{k-1}^j q_k(X_k^i \mid X_{1:k-1}^j) }}\)
        for $i = \range{L+1}{L+m_k^j}$.
      \item \textbf{delete} $\Obj^j$.
      \item Set $L \gets L+m_k^j$.
      \end{enumerate}
    \item Compute $\widehat Z_{\pi_k} = \widehat Z_{\pi_{k-1}}\times\left\{ \frac{1}{\Np}\sum_{j=1}^\Np \widehat\nu_{k-1}^j W_{k-1}^j \right\} \times \left\{ (\sum_{j=1}^\Np W_k^j)/(\sum_{j=1}^\Np W_{k-1}^j) \right\}.$
    \end{enumerate}
  \end{enumerate}
\end{algorithm}

\subsubsection{Nested \smc is \smc}\label{sec:supp:nsmc-is-smc}
The validity of Algorithm~\ref{alg:app:nested-smc} can be established by interpreting the
algorithm as a standard \smc procedure for a sequence of extended target distributions.
If $\widehat Z_{q_k}$ is computed deterministically, proper weighting (\ie, unbiasedness)
ensures that $\widehat Z_{q_k} = Z_{q_k}$ and it is evident that the algorithm reduces to a standard
\smc sampler. Hence, we consider the case when the normalising constant estimates $\widehat Z_{q_k}$ are random.

For $k = \range{1}{n+1}$,
let us introduce the random variable $U_{k-1}$ which encodes the complete internal state of the object
$\Obj$ generated by $\Obj = \Class(q_k(\cdot \mid x_{1:k-1}), M)$. Let
the distribution of $U_{k-1}$ be denoted as $\bar\psi_{k-1}^M (u_{k-1} \mid x_{1:k-1})$.
To put Algorithm~\ref{alg:app:nested-smc} into a standard (auxiliary) \smc framework,
we shall interpret steps \ref{step:app:nsmc:init-q}--\ref{step:app:nsmc:get-Z} of
Algorithm~\ref{alg:app:nested-smc} as being the last two steps carried out
during iteration $k-1$, rather than the first two steps carried out during iteration $k$.
This does not alter the algorithm \emph{per se}, but it results in that the resampling
step is conducted first at each iteration, which is typically the case for standard auxiliary
\smc formulations.

The estimator of the normalising constant is computable from the internal state of $\Obj$,
so that we can introduce a function $\tau_k$ such that $\widehat Z_{q_k} = \tau_k(U_{k-1})$.
Furthermore, note that the simulation of $X_k$ via $X_k = \Obj.\Simulate()$ is based
solely on the internal state $U_{k-1}$, and denote by $\bar\gamma_{k}^M(x_k \mid U_{k-1})$
the distribution of $X_k$.

\begin{lemma}
  \label{lem:app:proper-1}
  Assume that $\Class$ satisfies condition \asmpref{asmp:proper} in the main manuscript. Then,
  \begin{align*}
    \int \tau_k(u_{k-1}) \bar\gamma_k^M(x_k \mid u_{k-1}) \bar\psi_{k-1}^M(u_{k-1} \mid x_{1:k-1}) \myd u_{k-1}
    = q_k(x_k \mid x_{1:k-1}).
  \end{align*}
\end{lemma}
\begin{proof}
  The pair $(X_k, \tau_k(U_{k-1}))$ are properly weighted for $q_k$. Hence,
  for a measurable function $f$,
  \begin{multline*}
    \E[  f(X_k) \tau_k(U_{k-1}) \mid x_{1:k-1}]
    = \iint f(x_k)\tau_k(u_{k-1})
    \bar\gamma_k^M(x_k \mid u_{k-1}) \bar\psi_{k-1}^M(u_{k-1} \mid x_{1:k-1}) \myd u_{k-1}\myd x_k \\
    = Z_{k}(x_{1:k-1}) \int f(x_k) \bar q_k(x_k \mid x_{1:k-1}) \myd x_k
    = \int f(x_k) q_k(x_k \mid x_{1:k-1}) \myd x_k.
  \end{multline*}
  Since $f$ is arbitrary, the result follows.
\end{proof}

We can now define the sequence of (unnormalised) extended target distributions for the Nested SMC sampler as,
\begin{align*}
  \extUTarget_k(x_{1:k}, u_{0:k}) \eqdef
  \frac{ \tau_k(u_{k-1}) \bar\psi_k^M(u_k \mid x_{1:k}) \bar\gamma_k^M(x_k \mid u_{k-1}) }{ q_k(x_k \mid x_{1:k-1}) }
  \frac{\pi_k(x_{1:k}) }{ \pi_{k-1}(x_{1:k-1})}\extUTarget_{k-1}(x_{1:k-1}, u_{0:k-1}),
\end{align*}
and $\extUTarget_0(u_0) = \bar\psi_0^M(u_0)$.
We write $\setXU_k = \setX_k \times \setU_k$ for the domain of $\extUTarget_k$.

\begin{lemma}
  \label{lem:app:proper-2}
  Assume that $\Class$ satisfies condition \asmpref{asmp:proper} in the main manuscript. Then,
  \begin{align*}
    \int \tau_k(u_{k-1}) \bar\gamma_k^M(x_k \mid u_{k-1}) \extUTarget_{k-1}(x_{1:k-1}, u_{0:k-1}) \myd u_{0:k-1}
    = \pi_{k-1}(x_{1:k-1}) q_k(x_k \mid x_{1:k-1}).
  \end{align*}
\end{lemma}
\begin{proof}
  The proof follows by induction. At $k=1$, we have
  \(
  \int \tau_1(u_{0}) \bar\gamma_1^M(x_1 \mid u_{0}) \bar\psi_{0}^M(u_{0}) \myd u_{0}
  = q_1(x_1)
  \)
  by Lemma~\ref{lem:app:proper-1}. Hence, assume that the hypothesis holds
  for $k\geq 1$ and consider
  \begin{align*}
    \int &\tau_{k+1}(u_{k}) \bar\gamma_{k+1}^M(x_{k+1} \mid u_{k}) \extUTarget_{k}(x_{1:k}, u_{0:k}) \myd u_{0:k} \\
    &= \frac{\pi_k(x_{1:k}) }{ \pi_{k-1}(x_{1:k-1}) q_k(x_k \mid x_{1:k-1}) }   \\
  &\cdot \int \tau_{k+1}(u_{k}) \bar\gamma_{k+1}^M(x_{k+1} \mid u_{k})
    \tau_k(u_{k-1}) \bar\psi_k^M(u_k \mid x_{1:k}) \bar\gamma_k^M(x_k \mid u_{k-1}) 
  \extUTarget_{k-1}(x_{1:k-1}, u_{0:k-1}) \myd u_{0:k} \\
  &= \frac{\pi_k(x_{1:k})  \left( \int \tau_{k+1}(u_{k}) \bar\gamma_{k+1}^M(x_{k+1} \mid u_{k}) \bar\psi_k^M(u_k \mid x_{1:k}) \myd u_k \right)}{ \pi_{k-1}(x_{1:k-1}) q_k(x_k \mid x_{1:k-1})} \\
  & \cdot \int \tau_k(u_{k-1})  \bar\gamma_k^M(x_k \mid u_{k-1}) \extUTarget_{k-1}(x_{1:k-1}, u_{0:k-1}) \myd u_{0:k-1}  \\
  &= \frac{\pi_k(x_{1:k})   q_{k+1}(x_{k+1} \mid x_{1:k}) \pi_{k-1}(x_{1:k-1}) q_k(x_k \mid x_{1:k-1}) }{ \pi_{k-1}(x_{1:k-1}) q_k(x_k \mid x_{1:k-1})} = \pi_k(x_{1:k})   q_{k+1}(x_{k+1} \mid x_{1:k}),
  \end{align*}
  where the penultimate equality follows by applying Lemma~\ref{lem:app:proper-1} and
  the induction hypothesis to the two integrals, respectively.
\end{proof}

As a corollary to Lemma~\ref{lem:app:proper-2}, it follows that
\begin{align}
  \label{eq:app:correct-marginal}
  \int \extUTarget_k(x_{1:k}, u_{0:k}) \myd u_{0:k} = \pi_k(x_{1:k}).
\end{align}
Consequently, $\extUTarget_k$ is normalised by the same constant $Z_{\pi_k}$ as $\pi_k$, and by defining
$\extTarget_k(x_{1:k},u_{0:k}) \eqdef Z_{\pi_k}^{-1}\extUTarget_k(x_{1:k},u_{0:k})$ we obtain a probability distribution
which admits $\bar\pi_k$ as a marginal (note that $\extTarget_0 = \extUTarget_0$, which is normalised by construction).
This implies that we can use $\extTarget_k$ as a proxy for $\bar\pi_k$ in a Monte Carlo algorithm,
\ie, samples drawn from $\extTarget_k$ can be used to compute expectations \wrt $\bar\pi_k$.
This is precisely what Algorithm~\ref{alg:app:nested-smc} does; it is a standard auxiliary
\smc sampler for the (unnormalised) target sequence $\extUTarget_k$, $k=\range{0}{n}$, with adjustment
multiplier weights $\nu_{k-1}(x_{1:k-1}, \tau_k(u_{k-1}))$
and proposal distribution $\bar\gamma_k^M(x_k \mid u_{k-1}) \bar\psi_k^M(u_k \mid x_{1:k})$.
The (standard) weight function for this sampler is thus given by
\begin{align}
  \label{eq:app:apf-wght}
  W_k(x_{1:k}, u_{0:k}) \propto
  \frac{ \pi_k(x_{1:k}) }{  \pi_{k-1}(x_{1:k-1}) }
  \frac{\tau_k(u_{k-1}) }{\nu_{k-1}(x_{1:k-1}, \tau_k(u_{k-1})) q_k(x_k \mid x_{1:k-1}) },
\end{align}
which is the same as the expression on line~\ref{step:app:nsmcs:Wght} of Algorithm~\ref{alg:app:nested-smc}.

\subsubsection{Central Limit Theorem -- Proof of Theorem~\ref{thm:clt} in the Main Manuscript}\label{sec:supp:clt}
Now that we have established that Nested \smc is in fact a standard auxiliary \smc sampler, albeit on an extended
state space, we can reuse existing convergence results from the \smc literature; see \eg,
\citet{JohansenD:2008,DoucM:2008,DoucMO:2009,Chopin:2004} or the extensive textbook by \citet{DelMoral:2004}.

Here, in order to prove Theorem~\ref{thm:clt} of the main manuscript, we make use of the result for the auxiliary
\smc sampler by \citet{JohansenD:2008}, which
in turn is based on the central limit theorem by \citet{Chopin:2004}.
The technique used by \citet{JohansenD:2008} is to reinterpret (as detailed below)
the auxiliary \smc sampler as a \emph{sequential importance sampling and resampling} (SISR) particle filter, by introducing
the modified (unnormalised) target distribution
\begin{align}
  \label{eq:proof1:mod-target}
  \extUTarget'_k(x_{1:k}, u_{0:k}) \eqdef \nu_{k}(x_{1:k}, \tau_{k+1}(u_{k}))  \extUTarget_k(x_{1:k}, u_{0:k}).
\end{align}
The auxiliary \smc sampler described in the previous section can then be viewed as a SISR
algorithm for \eqref{eq:proof1:mod-target}. Indeed,
if we write
 \(
  \extQarg{k-1}{k} \eqdef \bar\psi_k^M(u_k \mid x_{1:k}) \bar\gamma_k^M(x_k \mid u_{k-1})
\)
for the joint proposal distribution of $(x_k, u_k)$, then the weight function for this SISR sampler is given by
\begin{align}
  \label{eq:proof2:mod-weight}
  \WghtSISR_k(x_{1:k}, u_{0:k}) & \eqdef \frac{ \extTarget'_k(x_{1:k}, u_{0:k}) }{  \extQarg{k-1}{k} \extTarget'_{k-1}(x_{1:k-1}, u_{0:k-1})} \nonumber\\
  &\propto \nu_k(x_{1:k}, \tau_{k+1}(u_k)) W_k(x_{1:k}, u_{0:k}),
\end{align}
where $W_k$ is defined in \eqref{eq:app:apf-wght}.
This weight expression thus accounts for both the importance weights and the adjustment multipliers of the auxiliary \smc sampler formulation.

Since this SISR algorithm does not target $\extTarget_k$
(and thus not $\target_k$) directly, we use an additional \is step to compute estimators
of expectations \wrt to $\target$. The proposal distribution for this \is procedure is
given by
\begin{align}
  \label{eq:proof1:is-proposal}
  \bar\Gamma_k(x_{1:k}, u_{0:k}) \eqdef
  \extQarg{k-1}{k}  \extTarget'_{k-1}(x_{1:k-1}, u_{0:k-1}).
\end{align}
Note that we obtain an approximation of \eqref{eq:proof1:is-proposal} after the propagation Step~\ref{step:app:nsmcs:Prop} of
Algorithm~\ref{alg:app:nested-smc}, but before the weighting step.
The resulting \is weights, for target distribution $\extTarget_k(x_{1:k}, u_{0:k})$
and with proposal distribution \eqref{eq:proof1:is-proposal}, are given by
\begin{align*}
  \frac{ \extTarget_k(x_{1:k}, u_{0:k}) }{    \bar\Gamma_k(x_{1:k}, u_{0:k})  }
  \defeq \WghtIS_k(x_{1:k}, u_{0:k}) \propto W_k(x_{1:k}, u_{0:k}).
\end{align*}
Hence, with $f: \setX_k \mapsto \reals^d$ being a
test function of interest we can estimate $\E_{\target_k}[f] = \E_{\extTarget_k}[f]$ (with obvious abuse of notation) by the
estimator
\begin{align}
  \label{eq:proof1:is-estimator}
  \sum_{i=1}^N \frac{W_k^i f(X_{1:k}^i)}{ \sum_{\ell=1}^N W_k^\ell},
\end{align}
which, again, is in agreement with Algorithm~\ref{alg:app:nested-smc}.

We have now reinterpreted the \nsmc algorithm; first as a standard \emph{auxiliary} \smc sampler,
and then further as a standard SISR method. Consequently, we are now in the position of directly
applying, \eg, the central limit theorem by \citet[Theorem~1]{Chopin:2004}. The conditions
and the statement of the theorem are reproduced here for clarity.

For any measurable function $f: \setXU_0 \mapsto \reals^d$,
let $\widetilde V_0^M(f) = \var_{\bar\psi_0^M}(f)$ and define, for any measurable function $f: \setXU_k \mapsto \reals^d$,
\begin{align*}
  \widetilde V_k^M(f) &= \widehat V_{k-1}^M (\E_{\extQ_k^M}[f]) + \E_{\extTarget'_{k-1}}[\var_{\extQ_k^M}(f)], & k&>0, \\
  V_k^M(f) &= \widetilde V_k^M( \WghtSISR_k (f-\E_{\extTarget'_k}[f] )),  & k&\geq 0, \\
  \widehat V_k^M(f) &= V_k^M(f) + \var_{\extTarget'_k}(f),   & k&\geq 0.
\end{align*}
Define recursively $\Phi_k$ to be the set of measurable functions $f: \setXU_k \mapsto \reals^d$ such that
there exists a $\delta > 0$ with ${\E_{\bar\Gamma_k}[ \| \WghtSISR_k f \|^{2+\delta} ] < \infty }$
and such that the function $(x_{1:k-1}, u_{0:k-1}) \mapsto \E_{\extQ_k^M}[\WghtSISR_k f]$ is in $\Phi_{k-1}$.
Furthermore, assume that the identity function $f \equiv 1$ belongs to $\Phi_k$ for each $k$. Then,
it follows by \citet[Theorem~1 and Lemma~A.1]{Chopin:2004} that
\begin{align}
  \label{eq:proof1:clt1}
  N^{1/2} \left( \sum_{i=1}^N \frac{1}{N} f(X_{1:k}^i, U_{0:k}^i) - \E_{\bar\Gamma_k}[f]) \right) \convD \N(0, \widetilde V_k^M(f)),
\end{align}
%where $\{X_{1:k}^i \}_{i=1}^M$ are generated by Algorithm~\ref{alg:app:nested-smc}
for any function $f$ such that the function
$(x_{1:k-1}, u_{0:k-1}) \mapsto \E_{\extQ_k^M}[f-\E_{\bar\Gamma_k}[f]]$ is in $\Phi_{k-1}$
and there exists a $\delta > 0$ such that ${\E_{\bar\Gamma_k}[ \|f \|^{2+\delta} ] < \infty }$.
% \[
% \int \| \WghtSISR_k(x_{1:k}, u_{0:k}) f(x_{1:k}, u_{0:k}) \|^{2+\delta}
% \extQarg{k-1}{k}  \extTarget_{k-1}'(x_{1:k-1}, u_{0:k-1}) \myd x_{1:k} \myd u_{0:k} < \infty,
% \]
The convergence in \eqref{eq:proof1:clt1} thus holds for the unweighted samples obtained
after the propagation Step~\ref{step:app:nsmcs:Prop} of
Algorithm~\ref{alg:app:nested-smc}, but before the weighting step.

To complete the proof, it remains to translate \eqref{eq:proof1:clt1} into a similar result for
the \is estimator \eqref{eq:proof1:is-estimator}. To this end we make use of \citet[Lemma~A.2]{Chopin:2004}
which is related to the \is correction step of the \smc algorithm.
Specifically, for a function $f : \setX_k \mapsto \reals^d$,
let $f^e : \setXU_k \mapsto \reals^d$ denote the extension of $f$ to $\setXU_k$,
defined by $f^e(x_{1:k}, u_{0:k}) = f(x_{1:k})$. Then, for any $f : \setX_k \mapsto \reals^d$ such that
the function $(x_{1:k-1}, u_{0:k-1}) \mapsto \E_{\extQ_k^M}[\WghtIS_k f^e]$ is in $\Phi_{k-1}$
and there exists a $\delta > 0$ such that ${\E_{\bar\Gamma_k}[ \| \WghtIS_k f^e \|^{2+\delta} ] < \infty }$,
we have
\begin{align*}
  N^{1/2}\left( \sum_{i=1}^N \frac{W_k^i f(X_{1:k}^i)}{ \sum_{\ell=1}^N W_k^\ell} - \target_k(f) \right)
  \convD \N(0, \Sigma_k^M(f)),
\end{align*}
where $\{(X_{1:k}^i, W_k^i) \}_{i=1}^M$ are generated by Algorithm~\ref{alg:app:nested-smc}
and
\(
  \Sigma_k^M(f) = \widetilde V_k^M( \WghtIS_k (f^e-\E_{\extTarget_k}[f^e] )).
\)

%  and $\convD$ denotes convergence
% in distribution.

\subsubsection{Nested \smc Generates Properly Weighted Samples -- Proof of Theorem~\ref{thm:proper} in the Main Manuscript}\label{sec:supp:proper}
In the previous two sections we showed that the \nsmc procedure is a valid inference algorithm for
$\bar\pi_n$. Next, we turn our attention to the modularity of the method and the validity of using
the algorithm as a component in another \nsmc sampler. Let us start by stating a more general
version of the backward simulator in Algorithm~\ref{alg:app:bs}. Clearly, if the forward \nsmc
procedure is fully adapted $W_k^i \equiv 1$, Algorithm~\ref{alg:app:bs} reduces to the backward
simulator stated in the main manuscript.

\begin{algorithm}
  \caption{Backward simulator}
  \label{alg:app:bs}
  \begin{enumerate}
  \item Draw $B_n$  from a categorical distribution with probabilities
      ${\displaystyle \frac{ W_n^j }{ \sum_{\ell=1}^\Np  W_n^\ell }}$ for $j = \range{1}{\Np}$.
  \item Set $\widetilde X_n = X_n^{B_n}$.
  \item \textbf{for $k=n-1$ to $1$}
    \begin{enumerate}
    \item Compute ${\displaystyle \widetilde W_k^j = W_k^j \frac{ \pi_n( (X_{1:k}^j, \widetilde X_{k+1:n}) ) }{ \pi_k(X_{1:k}^j) }}$ for $j = \range{1}{\Np}$.
    \item Draw $B_k$ from a categorical distribution with probabilities
      ${\displaystyle \frac{ \widetilde W_k^j }{ \sum_{\ell=1}^\Np  \widetilde W_k^\ell }}$ for $j = \range{1}{\Np}$.
    \item Set $\widetilde X_{k:n} = (X_k^{B_k}, \widetilde X_{k+1:n})$.
    \end{enumerate}
  \item \textbf{return $\widetilde X_{1:n}$}
  \end{enumerate}
\end{algorithm}

We will now show that the pair $(\widehat Z_{\pi_n}, \widetilde X_{1:n})$
generated by Algorithms~\ref{alg:app:nested-smc}~and~\ref{alg:app:bs} is properly weighted for $\pi_n(x_{1:n})$,
and thereby prove Theorem~\ref{thm:proper} in the main manuscript.

The proof is based on the \emph{particle Markov chain Monte Carlo} (PMCMC) construction \citep{andrieuDH2010particle}.
The idea used by \citet{andrieuDH2010particle} was to construct an extended target distribution,
incorporating all the random variables generated by an \smc sampler as auxiliary variables.
This opened up for using \smc approximations within \mcmc in a provably correct way; these seemingly
approximate methods simply correspond to standard \mcmc samplers for the (nonstandard) extended target
distribution. Here we will use the same technique to prove the proper weighing property
of the \nsmc procedure.

We start by introducing some additional notation for the auxiliary variables of the extended
target construction. 
While Algorithm~\ref{alg:app:nested-smc} is expressed using multinomial random variables $m_k^{1:N}$ in
the resampling step, it is more convenient for the sake of the proof to explicitly
introduce the \emph{ancestor indices} $\{A_k^i\}_{i=1}^N$; see \eg, \citet{andrieuDH2010particle}.
That is, $A_k^i$ is a categorical random variable on $\crange{1}{N}$, such that
$X_{1:k-1}^{A_k^i}$ is ancestor particle at iteration $k-1$ of particle $X_k^i$.
The resampling Step~\ref{step:app:nsmc:resample} of Algorithm~\ref{alg:app:nested-smc}
can then equivalently be expressed as: simulate independently $\{A_k^i\}_{i=1}^\Np$
from the categorical distribution with probabilities
\[
\frac{\widehat\nu_{k-1}^j W_{k-1}^j}{ \sum_{\ell=1}^\Np \widehat\nu_{k-1}^\ell W_{k-1}^\ell}.
\]
Let $\XX_k = \crange{ X_k^1 }{ X_k^N}$, $\UU_k = \crange{ U_k^1 }{ U_k^N}$, and
$\AA_k = \crange{ A_k^1 }{ A_k^N}$, denote all the particles, internal states of the proposals, and
ancestor indices, respectively, generated at iteration $k$ of the \nsmc algorithm.  We can then
write down the joint distribution of all the random variables generated in executing
Algorithm~\ref{alg:app:nested-smc} (up to an irrelevant permutation of the particle indices) as,
\begin{align}
  \label{eq:proof2:psi}
  \dSMC(\xx_{1:n}, \uu_{0:n}, \aa_{1:n}) =
  \left\{ \prod_{i=1}^N \bar\psi_0^M(u_0^i) \right\}
  \prod_{k=1}^n \left\{ \prod_{i=1}^N  \frac{\widehat\nu_{k-1}^{a_k^i} W_{k-1}^{a_k^i}}{ \sum_{\ell=1}^\Np \widehat\nu_{k-1}^\ell W_{k-1}^\ell} \extQ_k^M(x_k^i, u_k^i \mid x_{1:k-1}^{a_{k}^i}, u_{k-1}^{a_k^i} )\right\},
\end{align}
where we interpret $\widehat\nu_{k}^i$ and $W_{k}^i$ as deterministic functions of $(x_{1:k}^i, u_{0:k}^i)$.

Let $B_n$ denote a random variable defined on $\crange{1}{\Np}$. 
The extended target distribution for \pmcmc samplers corresponding to \eqref{eq:proof2:psi}
is then given by
\begin{align}
  \label{eq:proof2:phi}
  \dPMCMC(\xx_{1:n}, \uu_{0:n}, \aa_{1:n}, b_n) \eqdef
  \frac{\widehat Z_{\pi_n}}{Z_{\pi_n}} \frac{W_n^{b_n}}{\sum_{\ell=1}^N W_n^\ell} \dSMC(\xx_{1:n}, \uu_{0:n}, \aa_{1:n}),
\end{align}
where $\widehat Z_{\pi_n}$ is a deterministic function of $(\xx_{1:n}, \uu_{0:n}, \aa_{1:n})$.
We know from \citet{andrieuDH2010particle} that $\dPMCMC$ is a probability distribution
which admits $\extTarget_n$ as its marginal distribution for $(X_{1:n}^{b_n}, U_{0:n}^{b_n})$.
Consequently, by \eqref{eq:app:correct-marginal} it follows that the marginal distribution of $X_{1:n}^{b_n}$
is $\target_n$. For later reference we define recursively
$b_{k-1} \eqdef a_{k}^{b_{k}}$ for $k = \range{1}{n}$, the particle indices for the trajectory
obtained by tracing backward the genealogy of the $b_n$'th particle at iteration~$n$.

We now turn our attention to the backward simulator in Algorithm~\ref{alg:app:bs}.
Backward simulation has indeed been used in the context of \pmcmc, see \eg \citet{Whiteley:2010,LindstenS:2013,LindstenJS:2014}.
The strategy used for combining \pmcmc with backward simulation is to show that
each step of the backward sampler corresponds to a \emph{partially collapsed} Gibbs sampling step
for the extended target distribution $\dPMCMC$. This implies that the backward sampler leaves
$\dPMCMC$ invariant.

We use the same approach here, but we need to be careful in how we apply the existing results,
since the \pmcmc distribution $\dPMCMC$ is defined \wrt to $\extTarget_n$, whereas
the backward simulator of Algorithm~\ref{alg:app:bs} works with
the original target distribution $\target_n$. Nevertheless, from the proof of Lemma~1 by \citet{LindstenJS:2014}
it follows that we can write the following collapsed conditional distribution of $\dPMCMC$ as:
\begin{align}
  \label{eq:proof2:bs-i1}
  \dPMCMC( b_k, u_{k:n}^{b_{k:n}} &\mid  \xx_{1:k}, \uu_{0:k-1}, \aa_{1:k}, x_{k+1:n}^{b_{k+1:n}}, b_{k+1:n})\nonumber \\
  &\propto W_k^{b_k} \frac{ \extUTarget_n( \{ x_{1:k}^{b_k}, x_{k+1:n}^{b_{k+1:n}} \},\{u_{0:k-1}^{a_{k}^{b_k}}, u_{k:n}^{b_{k:n}} \} )  }{ \extUTarget_k( x_{1:k}^{b_k}, \{u_{0:k-1}^{a_{k}^{b_k}}, u_{k}^{b_{k}} \} ) }
  \bar\psi_k^M(u_k^{b_k} \mid x_{1:k}^{b_k}).
\end{align}
To simplify this expression, consider,
\begin{align}
  \notag
  \frac{ \extUTarget_n(x_{1:n}, u_{0:n}) }{  \extUTarget_k(x_{1:k}, u_{0:k}) }
  &= \prod_{s=k+1}^n \left\{ \frac{ \tau_s(u_{s-1}) \bar\psi^M_{s}(u_{s} \mid x_{1:s}) \bar\gamma^M_{s}(x_s \mid u_{s-1}) }{ q_s(x_s \mid x_{1:s-1}) }  \frac{ \utarget_s(x_{1:s}) }{ \utarget_{s-1}(x_{1:s-1}) }   \right\} \\
  \label{eq:proof2:bs-i2}
  &= \frac{\bar\psi_n^M(u_n\mid x_{1:n}) }{\bar\psi_k^M(u_k \mid x_{1:k}) }  \left\{\prod_{s=k+1}^n \frac{ \tau_s(u_{s-1}) \bar\psi^M_{s-1}(u_{s-1} \mid x_{1:s-1}) \bar\gamma^M_{s}(x_s \mid u_{s-1}) }{ q_s(x_s \mid x_{1:s-1}) }   \right\} \frac{ \utarget_n(x_{1:n}) }{ \utarget_k(x_{1:k}) }.
\end{align}
By Lemma~\ref{lem:app:proper-1} we know that each factor of the product (in brackets) on the second
line integrates to $1$ over $u_{s-1}$. Hence, plugging~\eqref{eq:proof2:bs-i2}
into~\eqref{eq:proof2:bs-i1} and integrating over $u_{k:n}^{b_{k:n}}$ yields
\begin{align*}
  \dPMCMC( b_k \mid  \xx_{1:k}, \uu_{0:k-1}, \aa_{1:k}, x_{k+1:n}^{b_{k+1:n}}, b_{k+1:n})
  \propto W_k^{b_k} \frac{ \utarget_n( \{ x_{1:k}^{b_k}, x_{k+1:n}^{b_{k+1:n}} \} )  }{ \utarget_k( x_{1:k}^{b_k} ) },
\end{align*}
which coincides with the expression used to simulate the index $B_k$ in Algorithm~\ref{alg:app:bs}.
Hence, simulation of $B_k$ indeed corresponds to a partially collapsed Gibbs sampling step for $\dPMCMC$
and it will thus leave $\dPMCMC$ invariant. (Note that, in comparison with the \pmcmc sampler derived by
\citet{LindstenJS:2014} we further marginalise over the variables $u_{k:n}^{b_{k:n}}$ which, however,
still results in a valid partially collapsed Gibbs step.)

We now have all the components needed to prove proper weighting of the combined \nsmc/backward
simulation procedure. For notational simplicity, we write
\begin{align*}
  \dBS{k}(b_k) = \dPMCMC( b_k \mid  \xx_{1:k}, \uu_{0:k-1}, \aa_{1:k}, x_{k+1:n}^{b_{k+1:n}}, b_{k+1:n}),
\end{align*}
for the distribution of $B_k$ in Algorithm~\ref{alg:app:bs}. 
Let $(\widehat Z_{\pi_n}, \widetilde X_{1:n})$ be
generated by Algorithms~\ref{alg:app:nested-smc}~and~\ref{alg:app:bs}.
Let $f$ be a measurable function and consider
\begin{align*}
  \E[ \widehat Z_{\pi_n} f( \widetilde X_{1:n}) ]
   &= \int \widehat Z_{\pi_n} f( X_{1:n}^{b'_{1:n}} ) \left \{ \prod_{k=1}^n   \dBS{k}(\myd b'_k) \right\}
   \dSMC(\myd ( \xx_{1:n}, \uu_{0:n}, \aa_{1:n})) \\
   &= Z_{\pi_n} \int f( X_{1:n}^{b'_{1:n}} ) \left \{ \prod_{k=1}^{n-1}   \dBS{k}(\myd b'_k) \right\}
   \dPMCMC(\myd ( \xx_{1:n}, \uu_{0:n}, \aa_{1:n}, b_n')),
\end{align*}
where, for the second equality, we have used the definition \eqref{eq:proof2:phi} and noted that
$\dBS{n}(b_n) = \frac{W_n^{b_n}}{\sum_{\ell=1}^N W_n^\ell} $.
However, by the invariance of $\dBS{k}$ \wrt $\dPMCMC$, it follows that
\begin{align*}
  \E[ \widehat Z_{\pi_n} f( \widetilde X_{1:n}) ]
   = Z_{\pi_n} \int f( X_{1:n}^{b_{1:n}} ) \dPMCMC(\myd ( \xx_{1:n}, \uu_{0:n}, \aa_{1:n}, b_n))
   = Z_{\pi_n} \target_n(f),
\end{align*}
which completes the proof.

%%% Local Variables:
%%% mode: latex
%%% TeX-master: "suppMain.tex"
%%% End:

%% file: suppNIS.tex
\subsection{Nested Sequential Importance Sampling}\label{sec:supp:nis}
Here we give the definition of the nested sequential importance sampler and we show that a special case of this is the importance sampling squared (IS$^2$) method by \citet{tranSPK2013importance}.

\subsubsection{Nested Sequential Importance Sampling}
We present a straightforward extension of the Nested \is class to a sequential \is version. Consider the following definition of the Nested SIS $\Class$:

\begin{enumerate}
  \setlength{\itemsep}{-1mm}
\item Algorithm~\ref{alg:nsmc:nested-sis} is executed at the construction of the object $\Obj[p] = \Class(\utarget_n, \Np)$,
  and $\Obj[p].\GetZ()$ returns the normalising constant estimate $\widehat Z_{\pi_n}$.
\item $\Obj[p].\Simulate()$ simulates a categorical random variable $B$ with $\Prb(B = i) = W_n^i / \sum_{\ell=1}^\Np W_n^\ell$
  and returns $X_{1:n}^B$.
\end{enumerate}

\begin{algorithm}[h]
  \caption{Nested SIS (\emph{all} $i$ \emph{for} $\range{1}{\Np}$)}
  \label{alg:nsmc:nested-sis}
    \begin{enumerate}
	  \item Initialise $\Obj^i = \Class(q_1(\cdot), M)$.
	  \item Set $\widehat Z_{q_1}^i = \Obj^i.\GetZ(), ~X_1^i = \Obj^i.\Simulate()$.
	  \item Set $W_1^i = {\displaystyle \frac{ \widehat Z_{q_1}^i \pi_1(X_1^i)}{ q_1( X_1^i) }}$.
      \item \textbf{delete} $\Obj^i$.
	  \item \emph{for} $k = 2$ \emph{to} $n$:
	  \begin{enumerate}
		\item Initialise $\Obj^i = \Class(q_k(\cdot \mid X_{1:k-1}^i), M)$.
		\item Set $\widehat Z_{q_k}^i = \Obj^i.\GetZ(), ~X_k^i = \Obj^i.\Simulate()$.
		\item Set $W_k^i = W_{k-1}^i{\displaystyle \frac{ \widehat Z_{q_k}^i \pi_k(X_k^i \mid X_{1:k-1}^i)}{ q_k( X_k^i \mid X_{1:k-1}^i) }}$.
        \item \textbf{delete} $\Obj^i$.
		\item Set $X_{1:k}^i \leftarrow (X_{1:k-1}^i, X_k^i)$
	  \end{enumerate}
	  \item Compute $\widehat Z_{\pi_n} = \frac{1}{\Np} \sum_{i=1}^\Np W_n^i.$
	  \end{enumerate}
\end{algorithm}

Note that we do not require that the procedure $\Class$ is identical for each individual proposal $q_k$, thus we have a flexibility in designing our algorithm as can be seen in the example in Section~\ref{sec:issquared}. We can motivate the algorithm in the same way as for Nested \is and similar theoretical results hold, \ie Nested \sis is properly weighted for $\pi_n$ and it admits $\bar \pi_n$ as a marginal.
%\begin{class}[h]
  %\caption{Nested SIS}
  %\label{alg:nsmc:nested-sis}
  %\begin{itemize}
  %\item[] $\mathsf{NSIS}$($\pi_n$, $\Np$) \hfill \emph{all} $i$ \emph{for} $1,\ldots, \Np$ 
	  %\begin{enumerate}
	  %\item Initialise $\Obj^i = \Class(q_1(\cdot), M)$.
	  %\item Set $\widehat Z_{q_1}^i = \Obj^i.\GetZ(), ~X_1^i = \Obj^i.\Simulate()$.
	  %\item Set $W_1^i = {\displaystyle \frac{ \widehat Z_{q_1}^i \pi_1(X_1^i)}{ q_1( X_1^i) }}$.
	  %\item \emph{for} $k = 2$ \emph{to} $n$:
	  %\begin{enumerate}
		%\item Initialise $\Obj^i = \Class(q_k(\cdot \mid X_{1:k-1}^i), M)$.
		%\item Set $\widehat Z_{q_k}^i = \Obj^i.\GetZ(), ~X_k^i = \Obj^i.\Simulate()$.
		%\item Set $W_k^i = W_{k-1}^i{\displaystyle \frac{ \widehat Z_{q_k}^i \pi_k(X_k^i \mid X_{1:k-1}^i)}{ q_k( X_k^i \mid X_{1:k-1}^i) }}$.
		%\item Set $X_{1:k}^i \leftarrow (X_{1:k-1}^i, X_k^i)$
	  %\end{enumerate}
	  %\item Compute $\widehat Z_{\pi_n} = \frac{1}{\Np} \sum_{i=1}^\Np W_n^i.$
	  %\end{enumerate}
  %\item[] $\Simulate()$
	  %\begin{enumerate}
	  %\item Simulate $B$ on $\crange{1}{N}$ with $\Prb(B = i) = \frac{ W_n^{i}}{ N \widehat Z_{\pi_n} }$.
	  %\item $\textbf{return}~X_{1:n}^B$.
	  %\end{enumerate}
  %\item[] $\GetZ()$
	  %\begin{enumerate}
	  %\item $\textbf{return}~ \widehat Z_{\pi_n}$.
	  %\end{enumerate}
  %\end{itemize}
%\end{class}

\subsubsection{Relation to IS$^2$}\label{sec:issquared}
Here we will show how IS$^2$, proposed by \citet{tranSPK2013importance}, can be viewed as a special case of Nested \sis. We are interested in approximating the posterior distribution of parameters $\theta$ given some observed values $y$
\begin{align*}
\bar\pi (\theta \mid y) \propto p(y \mid \theta) p(\theta).
\end{align*}
We assume that the data likelihood $p(y \mid \theta)$ can, by introducing a latent variable $x$, be computed as an integral
\begin{align*}
p(y \mid \theta) = \int p(y \mid x, \theta) p(x \mid \theta)\ \myd x.
\end{align*}
Now, let our target distribution in Nested \sis be $\bar \pi_2 (\theta, x) = \bar \pi_2(x \mid \theta) \bar \pi_1 (\theta) = \frac{p(y \mid x, \theta)p(x \mid \theta)}{p(y \mid \theta)} p(\theta)$. We set our proposal distributions to be
\begin{align*}
\bar q_1(\theta) &= g_{\text{IS}}(\theta), \\
\bar q_2(x \mid \theta) &= \frac{p(y \mid x, \theta)p(x \mid \theta)}{p(y \mid \theta)}.
\end{align*}
First, $\Class(q_1(\cdot),1)$ runs an exact sampler from the proposal $g_{\text{IS}}$. Then at iteration $k=2$ we let the nested procedure $\Class(q_2(\cdot \mid \theta^i),M)$ be a standard \is algorithm with proposal $h(x \mid y,\theta)$, giving us properly weighted samples for $q_2$. Putting all this together gives us samples $\theta^i$ distributed according to $g_{\text{IS}}(\theta)$ and weighted by
\begin{align}
W_2^i &\propto \frac{p(\theta^i)}{g_{\text{IS}}(\theta^i)} \cdot \frac{p(y \mid x^i, \theta^i) p(x^i \mid \theta^i) \frac{1}{M}\sum_{\ell=1}^M \frac{p(y \mid x^\ell, \theta^i) p(x^\ell \mid \theta^i)}{h(x^\ell \mid y, \theta^i)}}{p(y \mid x^i, \theta^i) p(x^i \mid \theta^i) } = \frac{\widehat p_M(y \mid \theta^i) p(\theta^i)}{g_{\text{IS}}(\theta^i)},\label{eq:issquared:w}
\end{align}
where $\widehat p_M(y \mid \theta^i) = M^{-1} \sum_{\ell=1}^M \frac{p(y \mid x^\ell, \theta^i) p(x^\ell \mid \theta^i)}{h(x^\ell \mid y, \theta^i)}$. Thus we obtain a Nested \sis method that is identical to the IS$^2$ algorithm proposed by \citet{tranSPK2013importance}.

%% file: suppExpts.tex
\subsection{Further Details on the Experiments}
\label{sec:supp:expts}
We provide some further details and results for the experiments presented in the main manuscript.
\subsubsection{Gaussian State Space Model}\label{sec:supp:expts:lgss}
We generate data from a synthetic $d$-dimensional (dim$(x_k) = d$) dynamical/spatio-temporal\footnote{Note that in a previous version this was erraneously stated as equivivalent to the Gaussian MRF we use for sequential inference. Thus this example actually illustrates a problem where we have a misspecified model. However, this misspecification does not lead to any discernible difference in the MSE results. This because the exact filtering marginals for the two different models (LGSS, GMRF) with the parameters chosen differs with orders of magnitudes much lower than the Monte Carlo errors.} model defined by
\begin{align*}
x_k \mid x_{k-1} &\sim \mathcal{N}(x_k; \mu_k(x_{k-1}), \Sigma), \\
y_k \mid x_k &\sim \mathcal{N}(y_k; x_k, \tau_\phi^{-1} I),
\end{align*}
where $\Sigma$ and $\mu_k$ are given as follows
\begin{align*}
\Sigma &= 
\begin{pmatrix}
\tau_\rho+\tau_\psi 	& -\tau_\psi			& 0 					& \cdots 				& \cdots 					& 0							& 0			\\
-\tau_\psi				& \tau_\rho+2\tau_\psi	& -\tau_\psi 			& 0 					& \cdots 					& 0 						& 0			 \\
0						& \ddots 				& \ddots 				& \ddots 				& \ddots 					& 0 						& 0			\\
\vdots					& \ddots 				& \ddots 				& \ddots 				& \ddots 					& 0							& 0		\\
\vdots					& \vdots				& \ddots 				& \ddots 				& \ddots 					& -\tau_\psi				& 0\\
0						& 0						& 0		 				& 0						& -\tau_\psi		 		& \tau_\rho+2\tau_\psi		& -\tau_\psi\\
0						& 0						& 0		 				& 0			 			& 0		 					& -\tau_\psi 				& \tau_\rho+\tau_\psi\\
\end{pmatrix}^{-1}, \\
\mu_k(x_{k-1}) &= a\tau_\rho \Sigma x_{k-1}.
\end{align*}
Alternatively, in a more standard state space model notation, we have
\begin{align*}
x_k &= A x_{k-1} + v_k, ~v_k \sim \mathcal{N}(0, Q), \\
y_k &= x_k + e_k, ~e_k \sim \mathcal{N}(0,R),
\end{align*}
where $A = a \tau_\rho \Sigma$, $Q = \Sigma$ and $R = \tau_\phi^{-1} I$. We assume that the parameters $\theta = (\tau_\psi, a, \tau_\rho, \tau_\phi) = (1, 0.5, 1, 10)$ are known.

To do inference with this generated data-set $\{y_k\}$ we propose to target the following slightly different model
\begin{align*}
p(x_{1:k},y_{1:k}) \propto \prod_{j=1}^k \boldsymbol\phi(x_j,y_j)\boldsymbol\rho(x_j) \boldsymbol\psi (x_j, x_{j-1}),
\end{align*}
where the observation potential $\boldsymbol\phi$ and interaction potentials $\boldsymbol\rho$ and $\boldsymbol\psi$ are given by
\begin{align*}
\boldsymbol\phi(x_k,y_k) &= \prod_{l=1}^d \phi_l (x_{k,l},y_{k,l}) = \prod_{l=1}^d e^{-\frac{\tau_\phi}{2} (x_{k,l} - y_{k,l})^2}, \\
\boldsymbol\psi(x_k) &= \prod_{l=2}^d \psi_l (x_{k,l},x_{k,l-1}) =  \prod_{l=2}^d e^{-\frac{\tau_\psi}{2} (x_{k,l} - x_{k,l-1})^2}, \\
\boldsymbol\rho (x_k, x_{k-1}) &= \prod_{l=1}^d  \rho_l (x_{k,l}, x_{k-1,l}) = \prod_{l=1}^d e^{-\frac{\tau_\rho}{2} (x_{k,l} - a x_{k-1,l})^2}.
\end{align*}
This can be visualised as a Gaussian rectangular ($d \times k$) lattice MRF, \ie it grows with ``time'' $k$. The goal is
to estimate the filtering distribution $p(x_k \mid y_{1:k})$. Note that this model has almost identical filtering marginals as the data generating distribution and leads to a simpler implementation of \nsmc and \stpf. %Furthermore, the above model can be
%rewritten (useful in the exact implementation) as a high-dimensional linear Gaussian state space
%model, 

Results (mean-squared-error, MSE) comparing \nsmc and \stpf for different settings of $N$ and $M$ can be found in the first row of Figure~\ref{fig:lgss:supp} and the second row displays the results when comparing \stpf to the \smc method by \citet{naessethls2014sequential} for equal computational budgets. We show median (over dimensions $d$) MSE for posterior marginal mean and variance estimates of the respective algorithms. True values are obtained using belief propagation. Note that setting $N=1$ in \stpf can be viewed as a special case of the \smc method by \citet{naessethls2014sequential}.

\begin{figure}[tb]
\centering
	\begin{tabular}{m{.5\textwidth} m{.5\textwidth}}
     \begin{center}{\scriptsize Median MSE for $\E[x_{k,\ell}]$}\end{center} & \begin{center}{\scriptsize Median MSE for $\var (x_{k,\ell})$}\end{center} \\
	 \includegraphics[width=.5\textwidth]{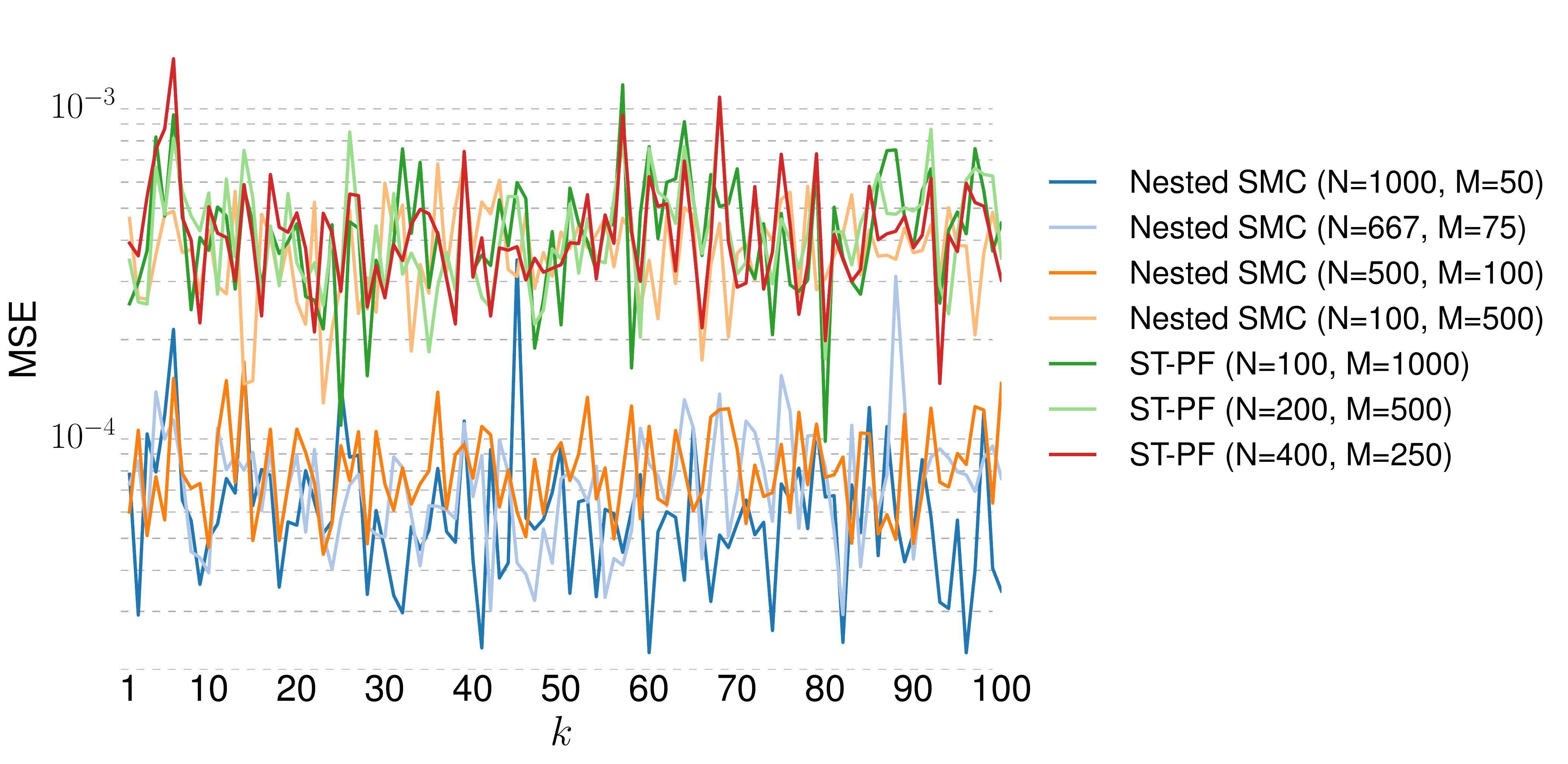} & \includegraphics[width=.5\textwidth]{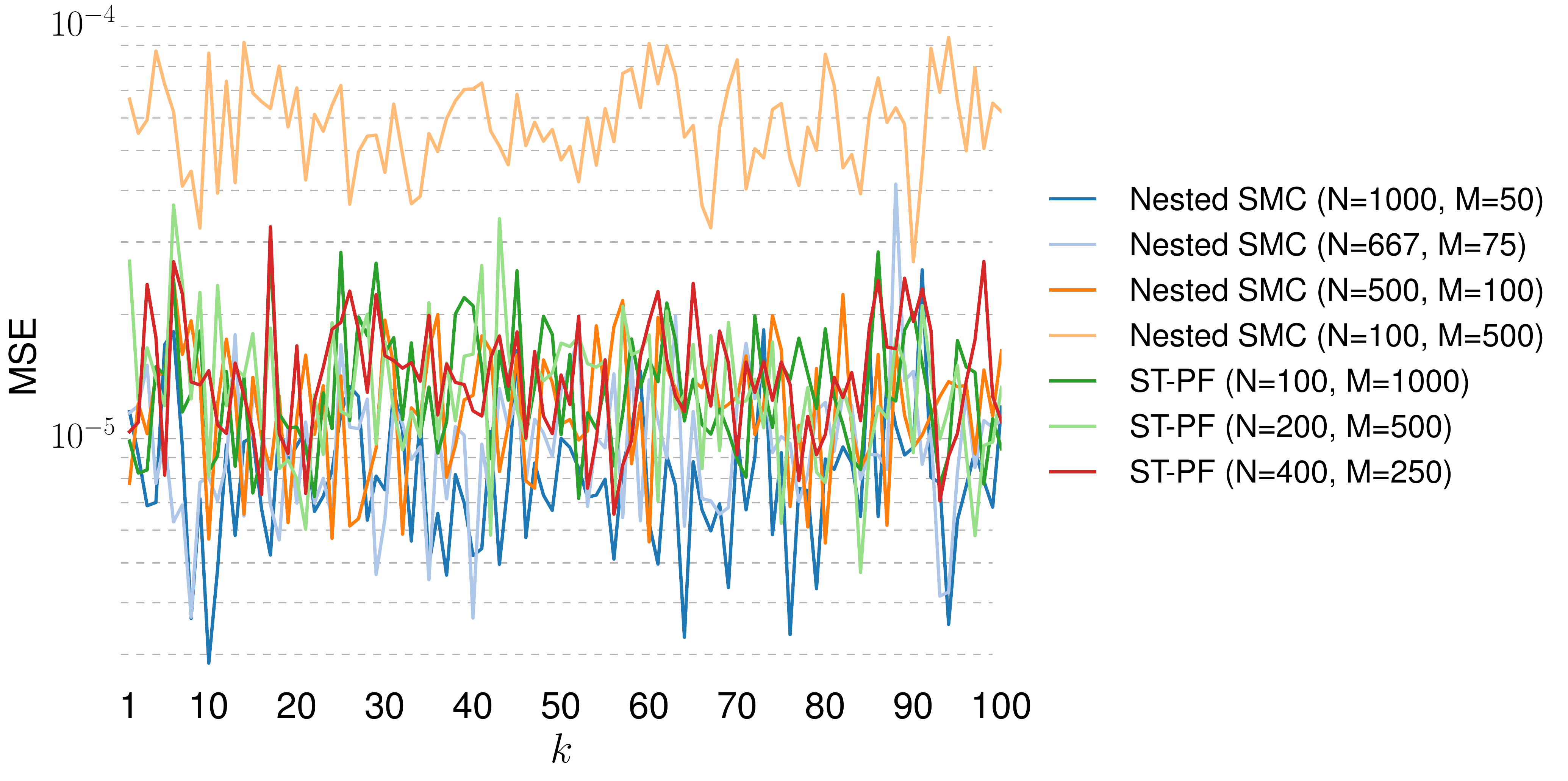} \\
     \begin{center}{\scriptsize Median MSE for $\E[x_{k,\ell}]$}\end{center} & \begin{center}{\scriptsize Median MSE for $\var (x_{k,\ell})$}\end{center} \\
	 \includegraphics[width=.5\textwidth]{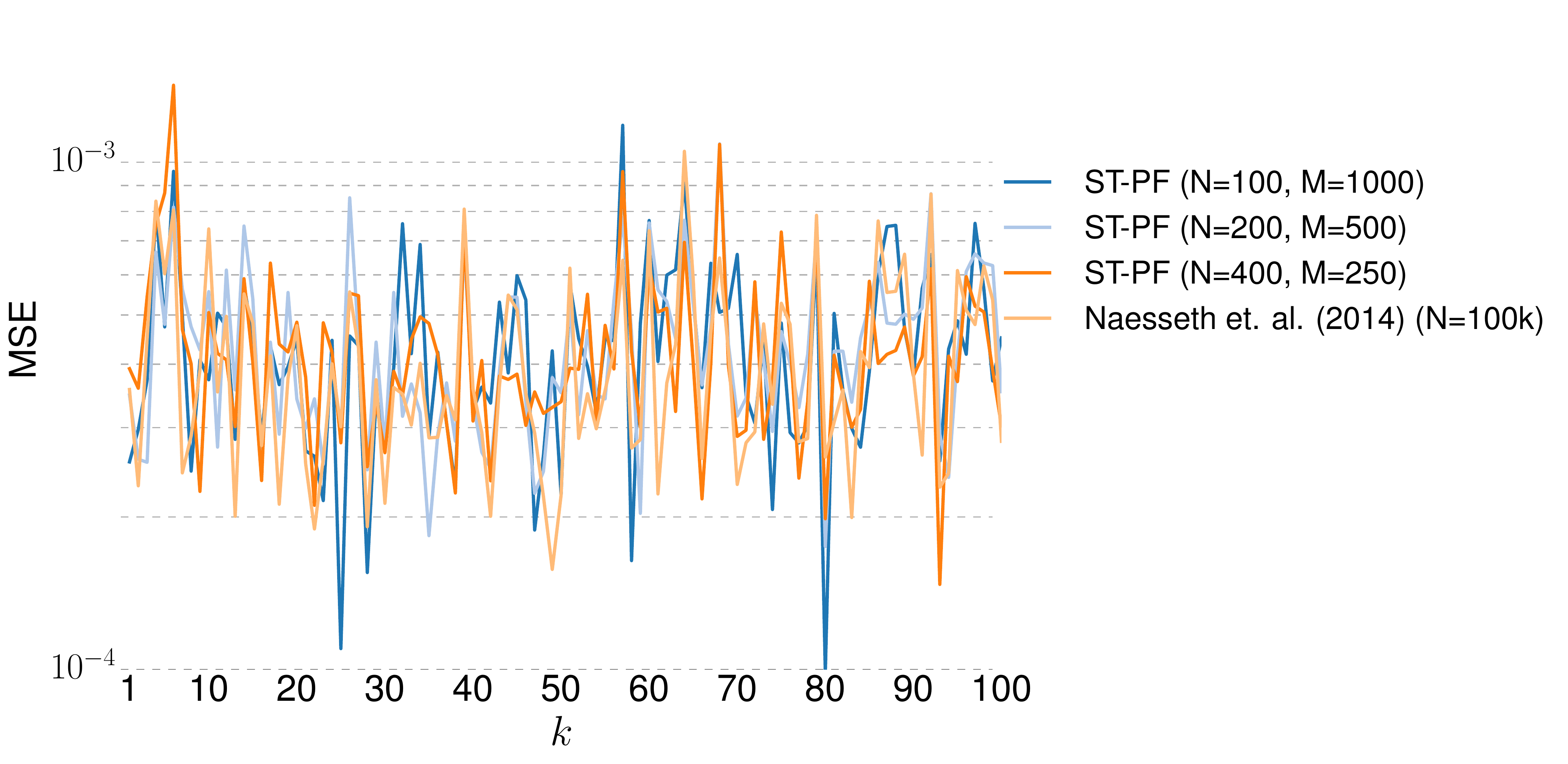} & \includegraphics[width=.5\textwidth]{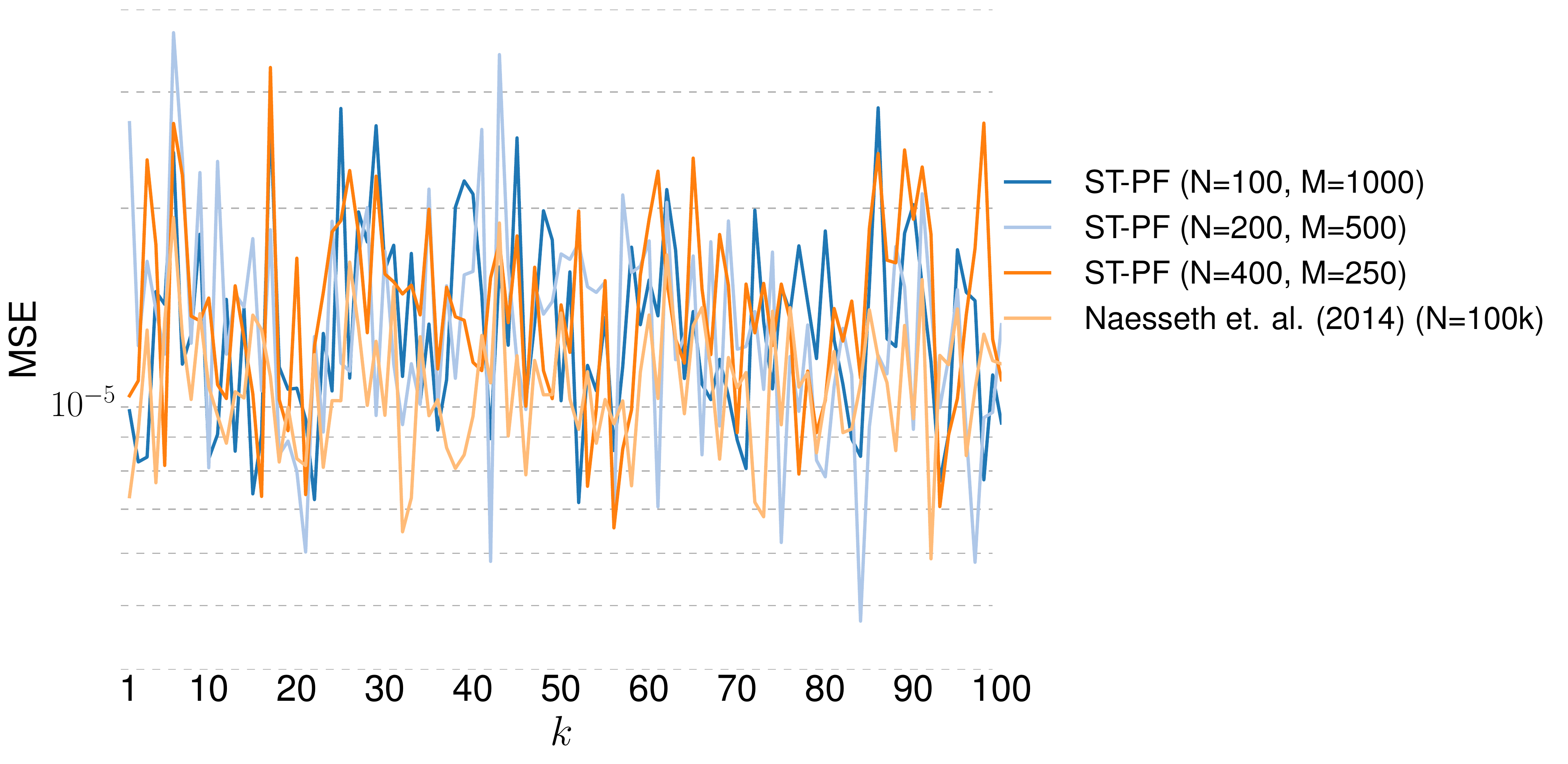}
	 \end{tabular}
	\caption{\emph{Top:} Comparisons for different settings of $N$ and $M$ on the $50$-dimensional \ssm. \emph{Bottom:} Illustrating the connection between \stpf and the \smc method by \citet{naessethls2014sequential}.}\label{fig:lgss:supp}
\end{figure}

\subsubsection{Spatio-Temporal Model -- Drought Detection}
We present the full model for drought detection in our notation, this is essentially the model by \citet{fuBLS2012drought} adapted for estimating the filtering distribution. The latent variables for each location on a finite world grid, $x_{k,i,j}$, are binary, \ie $0$ being normal state and $1$ being the abnormal (drought) state. Measurements, $y_{k,i,j}$, are available as real valued precipitation values in millimeters. The probabilistic model for filtering is given as,
%\ts{Provide the full model description in ``our notation''}
\begin{subequations}
\begin{align}
p(x_{1:k},y_{1:k}) \propto \prod_{n=1}^k \phi(x_n,y_n) \mathbb\rho(x_n) \mathbb\psi (x_n, x_{n-1}) ,
\end{align}
where
\begin{align}
\phi(x_k,y_k) &= \prod_{i=1}^{I} \prod_{j=1}^{J} \exp\left\{-\frac{1}{2 \sigma_{i,j}^2} \left( y_{k,i,j} - \mu_{\text{ab},i,j} x_{k,i,j} - \mu_{\text{norm},i,j} (1-x_{k,i,j}) \right)^2\right\}, \\
\rho (x_k) &= \prod_{i=1}^{I} \prod_{j=1}^{J} \exp\left\{ C_1 \left(\mathbbm{1}_{x_{k,i,j} = x_{k,i,j-1}} + \mathbbm{1}_{x_{k,i,j} = x_{k,i-1,j}} \right) \right\} , \label{eq:ind0}\\
\psi (x_k, x_{k-1}) &= \prod_{i=1}^{I} \prod_{j=1}^{J} \exp\left\{ C_2 \mathbbm{1}_{x_{k,i,j} = x_{k-1,i,j}} \right\}.
\end{align}
\end{subequations}
Here, $\mathbbm{1}$ is the indicator function, and with the convention that all expressions in \eqref{eq:ind0} that end up with index $0$ evalute to $0$. The parameters $C_1, C_2$ are set to $0.5, 3$ as in \citep{fuBLS2012drought}. Location based parameters $\sigma_{i,j}, \mu_{\text{ab},i,j}, \mu_{\text{norm},i,j}$ are estimated based on data from the CRU dataset with world precipitation data from years $1901-2012$. For the North America region we consider a $20 \times 30$ region with latitude $35-55^\circ N$ and longitude $90-120^\circ W$. For the Sahel region we consider a $24 \times 44$ region with latitude $6-30^\circ N$ and longitude $10^\circ W - 35^\circ E$. Note that for a few locations in Africa (Sahel region) the average yearly precipitation was constant. For these locations we simply set $\mu_{\text{norm},i,j}$ to be this value, $\mu_{\text{ab},i,j}=0$ and $\sigma_{i,j}^2$ to be the mean variance of all locations, thus this might have introduced some artifacts. Some representative results for the Sahel region are displayed in Figure~\ref{fig:sahel}.
\begin{figure}
\centering
	\begin{tabular}{m{.33\textwidth} m{.33\textwidth} m{.33\textwidth}}
	 \includegraphics[width=.33\textwidth]{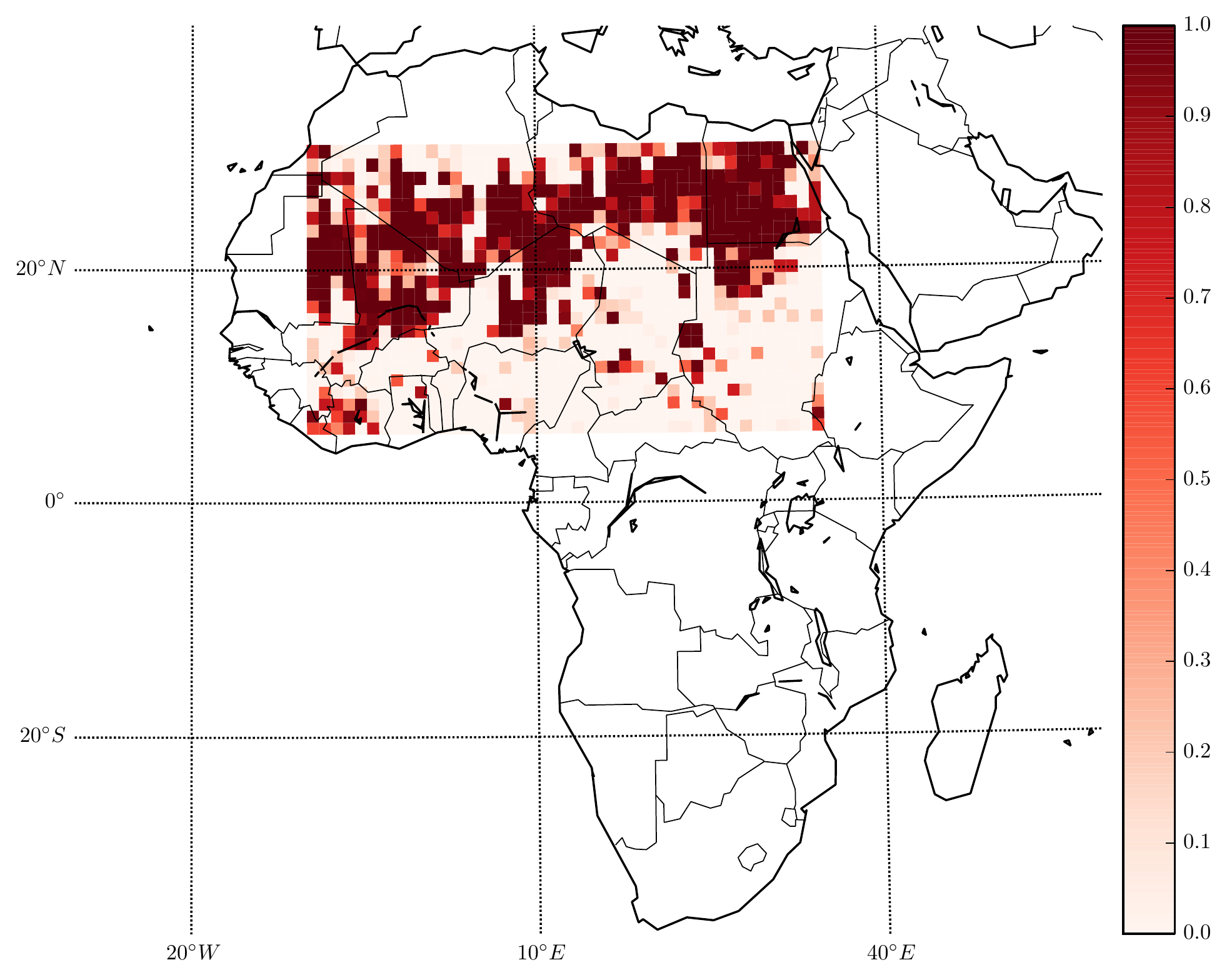} & \includegraphics[width=.33\textwidth]{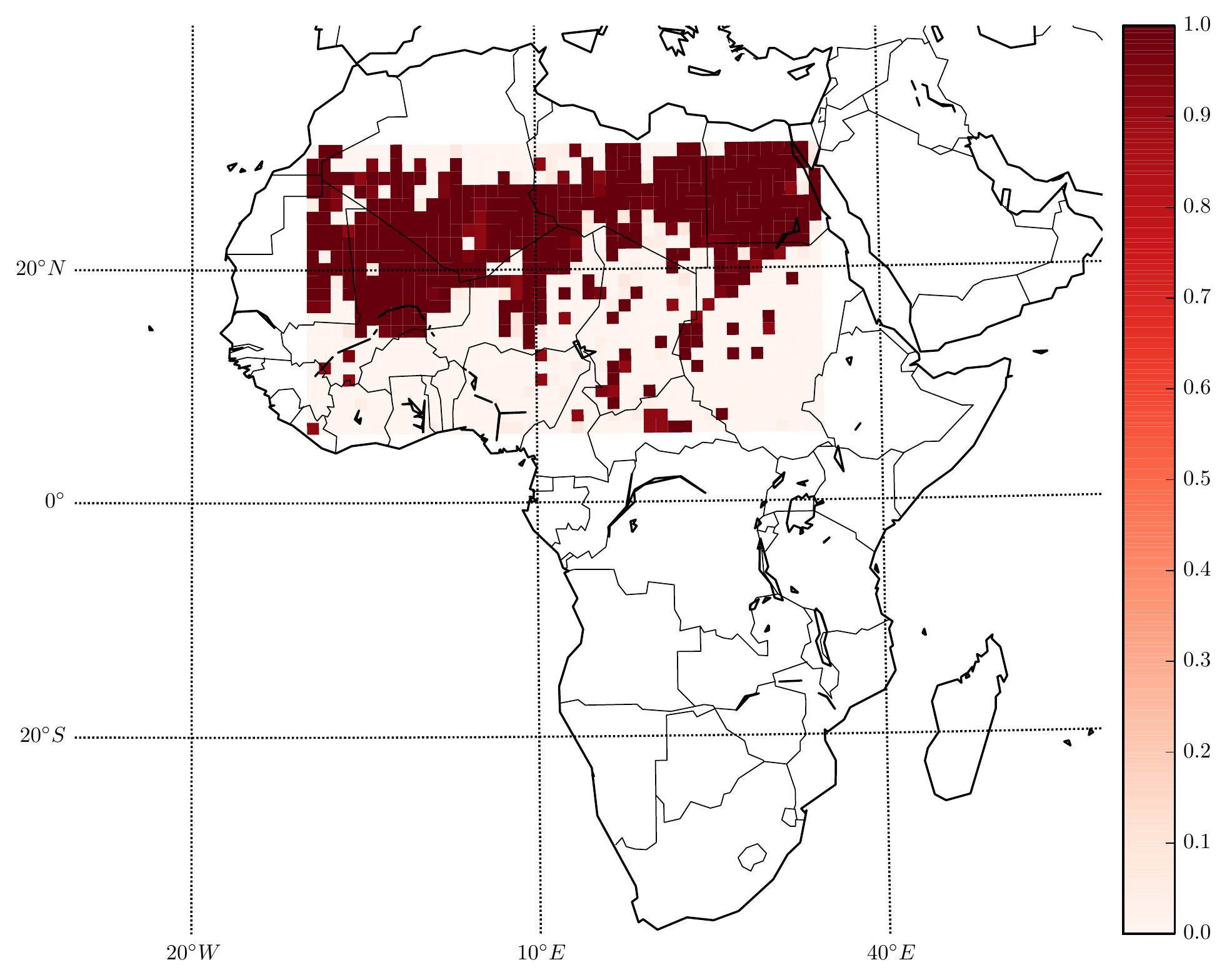} & \includegraphics[width=.33\textwidth]{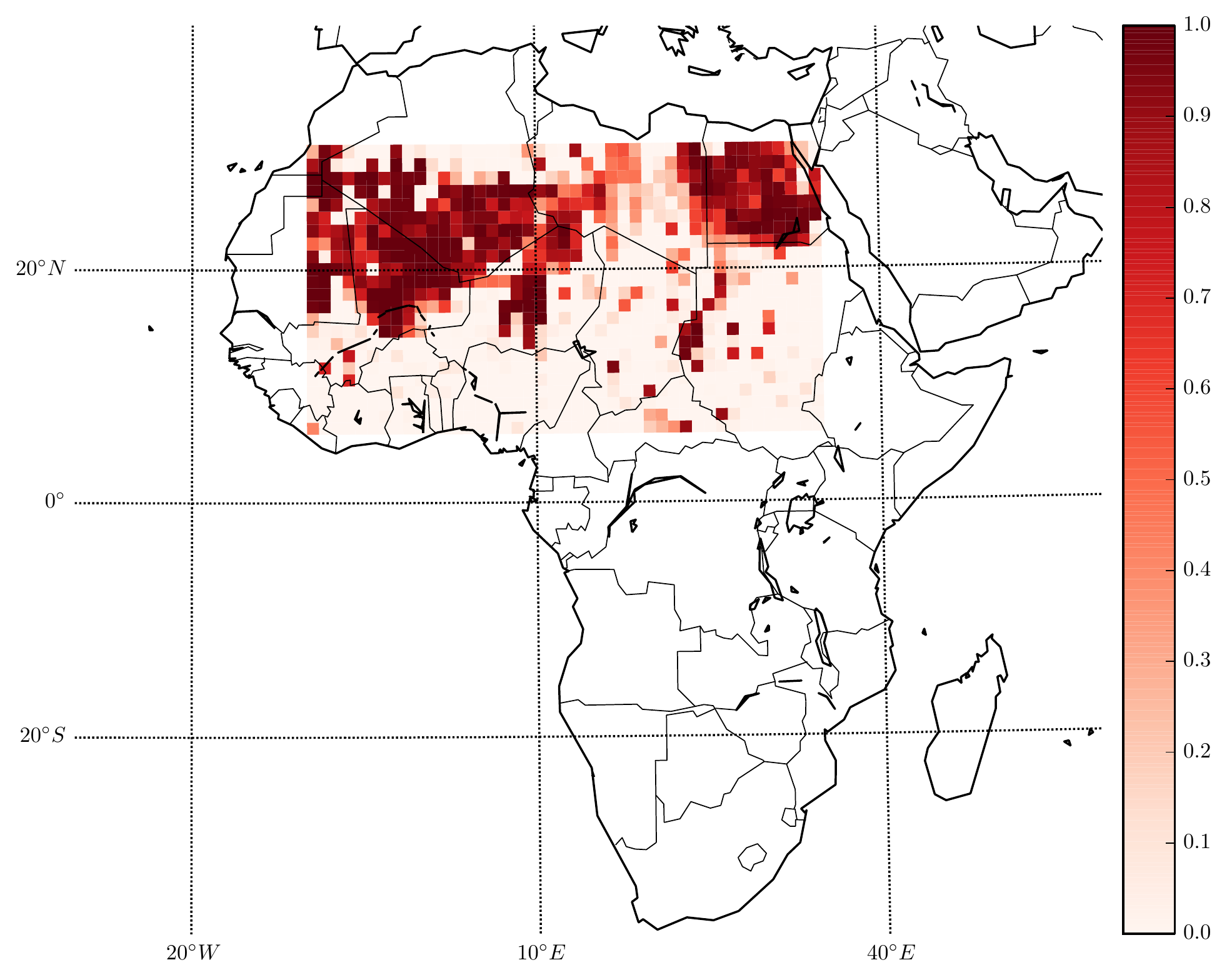} \\
	 \centering Sahel region $1986$ & \centering Sahel region $1987$ & \centering Sahel region $1988$
	 \end{tabular}
	\caption{Estimate of $\Prb(X_{k,i,j}=1 \mid y_{1:k})$ for all sites over a span of 3 years. All results for $\Np = 100, \Np_1 = \{30,40\}, \Np_2 = 20$.}\label{fig:sahel}
\end{figure}